\documentclass[12pt,english]{article}
\usepackage[T1]{fontenc}
\usepackage[latin9]{inputenc}
\usepackage{color}
\usepackage{babel}
\usepackage{verbatim}
\usepackage{mathtools}
\usepackage{amsmath}
\usepackage{bm}
\usepackage{amsthm}
\usepackage{amssymb}
\usepackage{graphicx}
\usepackage{setspace}
\usepackage{booktabs}
\usepackage{enumitem}
\usepackage[authoryear]{natbib}
\PassOptionsToPackage{normalem}{ulem}
\usepackage{ulem}
\usepackage[unicode=true,pdfusetitle,
bookmarks=true,bookmarksnumbered=false,bookmarksopen=false,
breaklinks=false,pdfborder={0 0 0},pdfborderstyle={},backref=false,colorlinks=true]
{hyperref}
\hypersetup{
	colorlinks=true,citecolor=blue,urlcolor=blue}

\makeatletter

\theoremstyle{plain}

\theoremstyle{plain}
\newtheorem{thm}{\protect\theoremname}
\theoremstyle{definition}

\theoremstyle{definition}
\newtheorem{dfn}{\protect\defname}

\setlist{itemsep=0pt,partopsep=0pt,topsep=0pt,parsep=0pt}

\allowdisplaybreaks

\@ifundefined{showcaptionsetup}{}{%
	\PassOptionsToPackage{caption=false}{subfig}}
\usepackage{subfig}
\makeatother

\providecommand{\examplename}{Example}
\providecommand{\lemmaname}{Lemma}
\providecommand{\theoremname}{Proposition}
\providecommand{\defname}{Definition}

\begin{document}
\title{\textbf{Bailouts and Redistribution}}
\author{Mikayel Sukiasyan\thanks{May 2021. I am grateful to my advisors Harald Uhlig, Mikhail Golosov and Simon Mongey for their guidance and advice.\newline Email: \texttt{msukiasyan@uchicago.edu}}}
\date{}
\maketitle

\begin{abstract}
What is the best macroprudential regulation when households differ in their exposure to profits from the financial sector? To answer the question, I study a real business cycle model with household heterogeneity and market incompleteness. In the model, shocks are amplified in states with high leverage, leading to lower investment. I consider the problem of a Ramsey planner who can finance transfers with a distortive tax on labor and levy taxes on the balance sheet components of experts. I show that the optimal tax on capital purchases is zero and the optimal policy relies mostly on a tax on deposit issuance. The latter redistributes between agents by affecting the equilibrium rate on deposits. The welfare gains from optimal policy are due to both redistribution and insurance and are larger the more unequal the initial distribution is. A simple tax rule that targets a level of leverage can achieve most of the welfare gains from optimal policy.
\end{abstract}

\newpage{}

\section{Introduction\label{sec:intro}}

The 2008-2009 financial crisis led to intense attention to the design of macroprudential policy. A set of actions taken by the Fed during the crisis -- characterized as bailouts -- brought forward debates on the necessity and design of not only ex-post policies but also preventive policies that act in the lead up to potential crises. Transfers of resources to highly-levered investors in capital are believed to avert deep recessions but are also feared to adversely affect incentives ex-ante. Much of the academic literature has examined these tradeoffs purely from the viewpoint of efficiency whereas the political debate has often touched upon issues of redistribution and equity. Who pays for bailouts and who stands to benefit? Is prudential policy justified if payers and beneficiaries do not coincide? These are some of the key questions that have shaped the discussion on Main Street versus Wall Street. 

In this paper, I study optimal macroprudential policy in an environment where households differ in their ownership of returns from risky investment. I consider a rich set of taxes that allow the policy-maker to finance transfers across the households as well as affect investment and balance sheet decisions. A crucial observation that becomes apparent in my setting is that prudential measures are also redistributive and this has to be taken into account when designing not only ex-ante but also ex-post policies. I find that the policy-maker can use the tax instruments to substantially improve the efficiency of the economy. As a consequence, the planner does not usually rely on prudential measures to prevent financial crises: under the optimal policies, high leverage and subsequent crises are mostly efficient.

I employ a real business cycle model to study these questions. The model is populated by two types of households: workers and experts. Two features of the model create interesting dynamics featuring episodes reminiscent of financial crises. First, markets are segmented so that only experts can trade and hold capital. Second, workers and experts are only allowed to trade a riskless claim between each other. In equilibrium, experts borrow from workers and invest in capital and, in doing so, expose themselves to elevated risk. When an exogenous shock destroys the value of their assets, experts must weather the impact only using their own net wealth since debt is not state-dependent. As a result, when experts' net wealth is depleted, they become more reluctant to invest as that becomes endogenously riskier for them. Deficient investment then spills over to workers through decreased wages in the future.

Transfers of resources from workers to experts when the latter are highly levered can lead to Pareto improvements. However, such transfers also amount to implicit insurance for the experts and encourage them to take riskier positions. Whether or not such ex-ante responses merit a policy response depends on the extent to which the policy improves risk sharing. Indeed, if risk sharing is perfect, there should be no disagreement between the agents on the amount of risk and investment that should be taken. Motivated by this observation, I assume that transfers cannot be levied lump-sum on workers and are instead financed by a distortive tax on labor. Further, I consider taxes on the issuance of riskless claims, as well as capital purchases to investigate the scope for prudential regulation. Assuming that tax revenues from these sources also flow to experts justifies their interpretation as macroprudential controls: an equivalence can be drawn between such taxes and quantity controls imposed on experts. But it turns out that, when households are heterogeneous, these policies also have redistributive implications! A positive tax on issuing the riskless claim, for instance, reduces the equilibrium rate of that claim and allows experts to effectively borrow more cheaply. This constitutes redistribution from workers to experts.

In the rich tax environment described above, I study the problem of a Ramsey planner operating under full commitment. I establish an important analytical result: the tax on capital purchases must be set to zero. The intuition for the result is that altering experts' investment decisions can only be beneficial due to its effect on equilibrium riskless rates and wages. The taxes on issuance of the claims and on labor, however, let the planner exhaust all gains from affecting those prices. Hence, the quantitative analysis focuses on the behavior of the remaining two tax instruments. There, I find that the tax on issuance of riskless claims finances most of the transfers under the optimal policy. The optimal policy significantly reduces the volatility of consumption of experts and, to a lesser extent, of workers. As a consequence of the reduction in risk, crises under optimal policies become somewhat less severe. Crucially, however, the policy-maker allows experts to build up leverage that is noticeably higher than that in equilibrium under no policies. When a large negative shock hits, ``ex-post'' policy kicks in to provide transfers to experts. It turns out that these transfers are predominantly financed by a positive tax on riskless claim issuance.

Building on lessons learned from studying optimal policies, I inspect the properties of simpler tax rules that capture the core features of the optimal ones. In particular, I consider a policy consisting of a simple tax on riskless claim issuance which is negative below and positive above a threshold of leverage. Such specification mirrors the fact that optimal policies aim to keep the leverage and with it, the wealth distribution, in an intermediate range for both redistributive and efficiency reasons.

Finally, I turn to welfare analysis and demonstrate that optimal policies attain most gains compared to those achieved under first-best allocations. I illustrate that gains are state-dependent and are larger the more unequal initial wealth distribution is. Moreover, the component of gains that is due to better risk sharing increases when the initial state of the economy is moved towards the high-leverage, high-risk region. I further show that the simple policy described above can attain a large fraction of the gains from optimal policies.

This paper is related to several strands of literature. First, I add to the literature on the macroeconomic implications of financial frictions. My model has broad similarities with those studied by \citet{gertler2010financial}, \citet{gertler2011model}, \citet{brunnermeier2014macroeconomic} and \citet{he2013intermediary}. The model I use is closely related to the macro-finance models studied by the latter two. Relative to their papers, I introduce a set of elements -- such as labor supply and persistent shocks -- that help cast the analysis in a more traditional RBC environment. Moreover, I use a Ramsey approach to systematically study optimal policy under explicit tax instruments. In doing so, I shed light on the extent to which policy can achieve efficiency and on issues of time-consistency. Similar to \citet{di2017uncertainty}, I emphasize that inability to contract on the aggregate state can be responsible for financial amplification. Like this paper, \citet{bocola2020risk} also examine risk sharing in a model with workers and entrepreneurs. They differ in that they assume contracts on the aggregate state can be written but limited commitment introduces pecuniary externalities. I, in contrast, maintain the assumption of market incompleteness and focus on normative analysis.

Second, my paper contributes to the literature on optimal macroprudential regulation in dynamic quantitative settings. I share with \citet{bianchi2016efficient} and \citet{bianchi2018optimal} the Ramsey approach to optimal regulation but they, like most others in that literature, study representative-agent environments. In these papers, levered investors borrow from unmodeled foreigners whose interests are not taken into account by the planner. \citet{bianchi2016efficient}, in particular, puts the question of prudential regulation at the center stage and finds that it has a very limited role. Notably, in that model, households that finance transfers to the levered firms also fully own the firms. As a result, the firms' response to the insurance offered by bailouts is efficient since they internalize that their owners are the ones paying for transfers. Instead, I consider a framework where payers and beneficiaries of transfers are different households and examine the influence of redistributive motives on optimal policy. In focusing on aspects of equity, I follow \citet{korinek2014redistributive} who were one of the first to invite attention to the redistributive conflict inherent in financial regulation.

Third, methodologically, this paper is tightly linked to the literature on Ramsey taxation under incomplete markets. \citet{aiyagari2002optimal} examined the problem of a government who needs to finance an exogenous stochastic stream of expenditures but can only issue one-period riskless bonds to the representative agent. In contrast, I assume the policy-maker has no need to finance expenditures of its own and instead is concerned with the welfare of heterogeneous agents. That the agents in my model can only trade a riskless claim, induces methodological similarity. Furthermore, in Section \ref{subsubsec:ramsey-crises}, I discuss how the emerging features of optimal taxes are shared between the two settings. A model of taxation under incomplete markets and with household heterogeneity is also studied by \citet{bhandari2018inequality}. They investigate optimal fiscal-monetary policy in a Bewley-Hugget-Aiyagari framework where the source of risk is labor income. Relative to them, I focus on limited heterogeneity and the importance of sharing investment risk. Limiting heterogeneity also allows me to solve the Ramsey planner's problem using global methods and explore how optimal policy varies over all dimensions of the state space.

Lastly, this paper is related to the literature on optimal fiscal policy in heterogeneous-agent economies such as \citet{heathcote2017optimal} and \citet{boar2020efficient}.Like them, I study how the planner must design taxes to achieve its redistributive goals. I differ in that I consider a model with financial frictions and aggregate risk. Further, although I also examine a space of taxes indexed by a few parameters, my main analysis takes a systematic Ramsey approach.

The paper proceeds as follows. Section \ref{sec:model} lays out the model and defines the tax instruments. Section \ref{sec:eqbm} describes the mechanisms at work in the model and shows how deep recessions can occur. Before turning to optimal policy analysis, Section \ref{sec:first-best} examines the properties of first-best allocations. Section \ref{sec:ramsey} contains the main analysis of optimal allocations, taxes as well as welfare implications. 

\section{Model\label{sec:model}}
In this section, I build a real business cycle model with two types of households, market segmentation and incomplete markets. I introduce in this setting tax instruments and consider the optimal planning problem. The model incorporates standard quantitative features without compromising the computational tractability of the optimal policy problem. However, maintaining tractability requires omitting a few features relative to the literature on financial crises.

\subsection{Setup\label{subsec:setup}}
Time is discreet. The economy is populated by two types of households -- workers and experts -- and final good producers. Expert households can hold existing and invest in new capital whereas the workers cannot. The two types are only allowed to trade a single one-period riskless claim denoted as deposits. This market incompleteness results in inefficient sharing of the aggregate risk. A spillover from experts onto workers operates through a mechanism where the former are occasionally too levered and consequently less willing to invest.

\subsubsection{Experts\label{subsubsec:experts}}
Experts are infinitely lived and maximize their lifetime expected utility. They begin each period with some holdings of capital $k_e$ and promised deposits $d_e$ to be repaid to workers. During the period, experts choose their consumption $c_e$, new promised deposits $d_e'$ and next period capital $k_e'$. Experts rent out their capital to firms at the rental rate $R$ per unit of capital and a $\delta$ fraction of the capital depreciates over the period. Experts can also invest $i$ units of the final good to build $\Phi(\frac{i}{k_e})k_e$ units of new capital. Finally, the experts can trade their end-of-period holdings of capital at price $q$ in the spot market. This is useful for us to derive a price of capital but note that no actual trades occur in this market since I assume all experts to be identical.

The policy-maker may impose two types of taxes on the experts: a deposit issuance tax $\tau_d$ and a capital purchases tax $\tau_k$. The proceeds from both are assumed to be rebated back to experts lump-sum. I will discuss this assumption in more detail when describing the government's budget constraint.

The dynamic program of the experts is given by

\begin{align}
    V_e(k_e, d_e, \textbf{S}) = &\max_{c_e,k_e',d_e',i} u_e(c_e) + \beta \mathbb{E}\left[V_e(\zeta' k_e', d_e', \textbf{S}')\right] \label{value-expert} \\
    s.t. \quad  c_e + \left(1 + \tau_k(\textbf{S})\right) q(\textbf{S}) k_e' + d_e = &\; \left(1-\tau_d(\textbf{S})\right)\frac{d_e'}{r(\textbf{S})} + R(\textbf{S}) k_e + (1-\delta)q(\textbf{S})k_e \label{eq:expert-budget} \\&+q(\textbf{S})\Phi\left(\frac{i}{ k_e}\right)k_e-i + T(\textbf{S}) \notag
    \\ \textbf{S}' = & \;\bm{\Gamma} (\textbf{S}) \label{trans-expert}
\end{align}

Here, $\textbf{S}$ is the vector of aggregate states and evolves according to $\bm{\Gamma}(\cdot)$. $r$ is the risk-free rate and $T$ denotes the lump-sum transfer to the experts. $\zeta'$ is an exogenous ``capital quality'' shock that allows to capture events during which the value of the assets held by experts deteriorates quickly. The details of the exogenous process will follow shortly. Finally, note that the expectation is taken over $\zeta'$ and $\textbf{S}'$.

The measurability of $r$  crucially restricts risk-sharing. $r$ is the rate that needs to be paid in the next period on deposits but it cannot depend on the realizations of next-period exogenous shocks.

The investment sub-problem is completely static and can be analyzed separately. The optimal investment solves the following problem
\begin{align*}
    \max_{i} q(\textbf{S}) \Phi\left(\frac{i}{k_e}\right)k_e - i
\end{align*}

Then the first-order condition combined with aggregation implies the pricing equation
\begin{align*}
    q(\textbf{S}) = \frac{1}{\Phi'\left(\frac{I(\textbf{S})}{K}\right)}
\end{align*}

\subsubsection{Workers\label{subsubsec:workers}}
Workers are also infinitely lived and maximize their lifetime expected utility. They begin the period with some deposit holdings $d_w$ to be paid by the experts. The workers then choose their consumption, labor $l_w$ and choose next period deposits $d_w'$. The policy-maker imposes a linear state-dependent tax $\tau_l$ on labor. The revenue from this tax is again transferred to the experts.

The workers' dynamic program is as follows 
\begin{align}
    V_w(d_w, \textbf{S}) = &\max_{c_w,l_w,d_w'} u_w(c_w,l_w) + \beta \mathbb{E}\left[V_w(d_w', \textbf{S}')\right] \label{value-worker}\\
    s.t. \quad c_w + \frac{d_w'}{r(\textbf{S})} = & \;d_w + \left(1 - \tau_l(\textbf{S})\right)W(\textbf{S}) l_w \label{eq:worker-budget}
    \\ \textbf{S}' = & \;\bm{\Gamma} (\textbf{S}) \label{trans-worker}
\end{align}

where $W$ denotes the wage rate.

Throughout, I assume  natural borrowing constraints for both agents.

\subsubsection{Final good producers\label{subsubsec:finalgood}}

Final good producers rent capital from the experts, hire labor from workers and produce final goods $Y$ using the constant returns to scale technology $F(Z,K,L)$ where $Z$ is the aggregate productivity. Clearly, the problem of the firm is static and has the following form
\begin{align*}
    \max_{Y,K,L} \quad &Y - W(\textbf{S})L - R(\textbf{S})K\\
    s.t. \quad & Y = F(Z,K,L)
\end{align*}

The first-order conditions of the problem imply the familiar conditions
\begin{align*}
    F_L(Z,K,L) &= W(\textbf{S})\\
    F_K(Z,K,L) &= R(\textbf{S})\\
\end{align*}

\subsubsection{Exogenous shock process\label{subsubsec:shocks}}

I assume that $\zeta$ takes one of two values in $\{\underline{\zeta},1\}$ and $Z$ takes values from the finite set $\{Z_L,..,Z_H\}$. Whenever $\zeta = 1$, $Z$ evolves according to a Markov process that approximates an AR(1) process. Whenever $\zeta = \underline{\zeta}$, $Z$ takes the value $\underline{Z}$. $Z$ is drawn from its marginal ergodic distribution when $\zeta$ switches from $\underline{\zeta}$ to $1$.

Following the literature on disaster risk (e.g. \citet{gourio2012disaster}), the purpose of this structure is to model a rare extreme shock that combines deterioration in experts' asset value with drop in productivity.

\subsubsection{Government\label{subsubsec:gov}}

In each period, the government collects revenues from taxes on labor, deposit issuance, capital purchases and rebates all proceeds to the experts, i.e. the budget is balanced period-by-period:
\begin{align}
    \tau_l(\textbf{S}) W(\textbf{S}) L(\textbf{S}) + \tau_d(\textbf{S}) \frac{D'(\textbf{S})}{r(\textbf{S})} + \tau_k(\textbf{S}) q(\textbf{S}) K'(\textbf{S}) = T(\textbf{S}) \label{eq:gov-budget}
\end{align}

where upper-case $D'$, $K$ and $L$ denote the aggregate quantities of the corresponding variables.

The assumption that all revenues are rebated to the experts merits a comment. As we will see, if the planner is allowed to choose state-dependent lump-sum transfers for both agents, the policy-maker can achieve first-best outcomes. Moreover, assuming that the proceeds go to the experts rather than to the workers allows to interpret the taxes as macroprudential instruments. Under this interpretation, the labor tax is used to finance a ``bailout'' to the experts while the taxes on deposit issuance and capital purchases act as state-dependent \textit{quantity controls}. Indeed, note that imposing quantity caps on deposit issuance and capital purchases is an equivalent formulation where the equilibrium prices coincide with the pre-tax prices in the current formulation.

The second assumption that deserves attention is that experts perceive the transfer to be lump-sum. The literature on macroprudential policies has commonly assumed that financial intermediaries perceive bailouts to be related to their individual balance sheet conditions. \citet{bianchi2016efficient}, for instance, considers ``systemic bailout'' policies for which transfers are proportional to the individual debt of the financial firm while the proportionality factor itself is restricted to only depend on the aggregate state (hence, the policy is called systemic). In such models, the dependence of transfers on individual conditions distorts the relevant Euler equations for the recipients of the transfers which in turn necessitates prudential controls to offset distortiona. In this paper, I abstract away from such dependence in order to focus on the role of prudential policies when a redistributive motive is present.

\subsection{Equilibrium with taxes\label{subsec:eqdef}}

We're now ready to define the recursive competitive equilibrium with given tax policies. 

\begin{dfn}
\label{dfn:eqbm} Let $\textbf{S} = \{K,D,Z,\zeta_0\}$. A recursive competitive equilibrium with tax policies $\tau_l(\textbf{S})$, $\tau_d(\textbf{S})$ and $\tau_k(\textbf{S})$ is a collection of value functions $V_e$, $V_w$, policy functions $c_e$, $k_e'$, $d_e'$, $i$, $c_w$, $d_w'$, price functions $W$, $R$, $r$, $q$ as well as a law of motion for the state, $\bm{\Gamma}$, such that
    \begin{enumerate}
        \item Policy functions solve household utility maximization problems given prices, the laws of motion and continuation values.
        \item Value functions solve the corresponding functional equations.
        \item Firms maximize their profits given prices.
        \item Aggregation holds and markets clear
        \begin{align*}
            \zeta' k_e'(K, D,\textbf{S}) = \zeta' K \left( (1-\delta)  + \Phi\left(\frac{I}{K}\right) \right)  &= K'(\textbf{S})\\
            d_e'(K, D,\textbf{S}) = d_w'(D,\textbf{S}) &= D'(\textbf{S})\\
            l_w(D,\textbf{S}) &= L(\textbf{S})\\
            i(K, D,\textbf{S}) + c_w(D,\textbf{S}) + c_e(K,D,\textbf{S}) = I(\textbf{S}) + C_w(\textbf{S}) + C_e(\textbf{S}) &= Y(\textbf{S})
        \end{align*}
        \item The law of motion of the aggregate state induced by the policy functions and the evolution of exogenous shocks is consistent with that used by the households.
    \end{enumerate}
\end{dfn}

As will become apparent later, the optimal plans arising from the Ramsey problem do not generically belong to the class of tax policies assumed in this definition. An appropriately augmented definition will be presented in Section \ref{subsec:ramsey-rec}.

\subsection{Planner's objective\label{subsec:ramseydef}}

Given an initial state $\textbf{S}_0=\{K_0,D_0,Z_0,\zeta_0\}$, the Ramsey planner ranks allocations according to the following objective
\begin{align}
    \mathbb{E}_0 \left[ \sum_{t = 0}^\infty \beta^t \left(\lambda u_e(C_{et}) + (1-\lambda) u_w(C_{wt},L_t) \right) \right]
\end{align}

where $\lambda$ is the Pareto weight assigned to the experts.

The value of $\lambda$ influences the quantitative but not the qualitative properties of the optimal plans. Note that, until now, the implicit assumption has been that the measures of workers and experts are the same. In Appendix \ref{app:non-unit}, I will show that any equilibrium with equal (unit) measures is also an equilibrium of an extended model with different measures provided that a scaling parameter for labor disutility is adjusted. In light of this, the Pareto weight parameter $\lambda$ corresponding to the objective of a purely utilitarian planner should be thought to be a function of the assumed measures and utility parameters. 

\section{Equilibrium characterization\label{sec:eqbm}}

It is useful to first consider the behavior of the model in the absence of any tax policy. The competitive equilibrium of this baseline model features an inefficiency due to insufficient risk sharing between the two agents. Due to market segmentation, this inefficiency also manifests itself in insufficient amount of investment when the experts are highly levered. The reason is that when experts' own wealth is low and capital purchases are mostly financed using household deposits, any shock to the value of the assets has a multiplicative effect on the wealth of experts. Consequently, if experts choose higher leverage they expose themselves to higher risk. This, in turn, reduces their willingness to issue deposits and invest. Reduction of investment due to inefficient risk sharing has an effect on workers due to lower wages. In such regions of the state space, it is desirable to transfer resources to experts, particularly when a bad shock has realized. This not only helps smooth their consumption but also encourages more investment.

Sharing of risk between the two types of agents is inefficient not only when experts have low wealth and therefore are highly levered, but also when workers have relatively low wealth. Whereas experts are exposed to the aggregate risk through the return on capital, workers are exposed via their labor income. Workers use savings in deposits to insure against these fluctuations and smooth consumption. When savings are low, workers are exposed to comparatively more risk then experts. In such a scenario, transfers to workers when a bad shock hits are what is desirable.

In order to describe the behavior of the model quantitatively, we will start by parametrizing the model using standard real business cycle parameters.

\subsection{Parametrization\label{subsec:calib}}

I choose a standard separable form for the utility function of the worker with CRRA utility of consumption for both agents
\begin{align*}
     u_w(c,l) &= \frac{c^{1-\gamma}}{1-\gamma} - \psi \frac{l^{1+\nu}}{1+\nu}\\ 
     u_e(c) &= \frac{c^{1-\gamma}}{1-\gamma}
\end{align*}

where $\gamma$ is the risk aversion parameter and $\nu$ is the inverse Frisch elasticity of labor supply.

The parametric form for the capital production function is as follows
\begin{align*}
    \Phi(x) &= \frac{\left(A (x - \delta) + 1 \right)^ {\xi} - 1}{A\xi} + \delta
\end{align*}

where $\xi \in (0, 1]$ relates to the elasticity of the capital price $q$ with respect to the investment-to-capital ratio. Note that this functional form satisfies $\Phi(\delta) = \delta$, $\Phi'(\delta) = 1$. That is, around $\delta$, the capital production function features no adjustment costs locally.

The production function is given by the standard Cobb-Douglas form
\begin{align*}
    F(Z,K,L) = e^Z K ^ \alpha L (1 - \alpha)
\end{align*}

where $\alpha$ is the capital share.

Conditional on $\zeta = \underline{\zeta}$, log productivity follows a Markov process that approximates the following AR(1) process on the grid $\{Z_L,\ldots,Z_H\}$ based on the \citet{cooley1995frontiers} method
\begin{align*}
    Z' = \rho_Z Z + \varepsilon_Z, \quad \varepsilon_Z \sim N(0, \sigma_\varepsilon^2)
\end{align*}

The reference period for the model is a quarter. Table \ref{tab:calib} lists all parameter values. Table \ref{tab:calib-targets} shows several characteristics of the equilibrium under this parametrization.
        
\begin{table}[h]
\begin{center}
    \begin{tabular}{ ccc } 
     \toprule
     \textbf{Parameter} & \textbf{Meaning} & \textbf{Value}  \\
     \midrule
     $\beta$ & Discount factor & 0.97\\ 

     $\gamma$ & Risk aversion & 2.00\\ 

     $\nu$ & Inverse Frisch & 1.00\\ 
     
     $\psi$ & Labor scaling parameter & 1.00\\ 

     $\xi$ & Capital production parameter & 0.80\\ 

     $A$ & Capital production parameter & 1.25\\ 

     $\alpha$ & Capital share & 0.36\\ 

     $\delta$ & Depreciation rate & 0.025\\ 

     $\rho_Z$ & Persistence of log productivity & 0.90\\ 

     $\sigma_\varepsilon$ & Std dev of innovation & 0.08\\ 
     
     $\underline{Z}$ & Disaster log productivity & -0.15\\ 

     $\underline{\zeta}$ & Capital quality in bad state & 0.95\\ 

     $\pi_{\underline{\zeta}}$ & Prob. of switching to $\underline{\zeta}$ & 0.5\%\\ 

     $\pi_{\overline{\zeta}}$ & Prob. of switching to $\underline{\zeta}$ & 34\%\\ 
     $\lambda$ & Planner's Pareto weight & 0.01\\ 
     \bottomrule 
    \end{tabular}
    \caption{\label{tab:calib} Parameters of the model.}
\end{center}
\end{table}

\subsection{Numerical solution\label{subsec:num-sol}}

Since the model features an inefficiency, standard value function iteration methods are not readily applicable. Instead, to find a global solution, I obtain a system of equilibrium conditions using the first-order conditions of the agents combined with aggregation, as well as market clearing and the law of motion of the state vector. I approximate the equilibrium objects with cubic splines and use a combination of iterative and Newton-type methods to find a solution. Appendix \ref{app:numerical-method} contains the details.

\subsection{Equilibrium objects\label{subsec:eqbm-objects}}

Figure \ref{fig:eqbm_policies} plots equilibrium investment and deposit rate over a few slices of the state space where the capital quality shock is active. The horizontal axis shows the key state variable -- deposits over capital -- that affects the experts' willingness to issue deposits and invest. Panel (b) shows that the equilibrium riskless rate decreases as deposits rise for a fixed level of capital. As deposits increase, workers become less willing to save and want to consume instead, all else equal. However, as the net worth of experts declines, they become hesitant to take on more debt and invest. This second effect is stronger and pushes the equilibrium rate down. Panel (a) shows that investment varies substantially with wealth inequality. 

\begin{figure}[h]
\begin{center}
    \includegraphics[scale = 0.70]{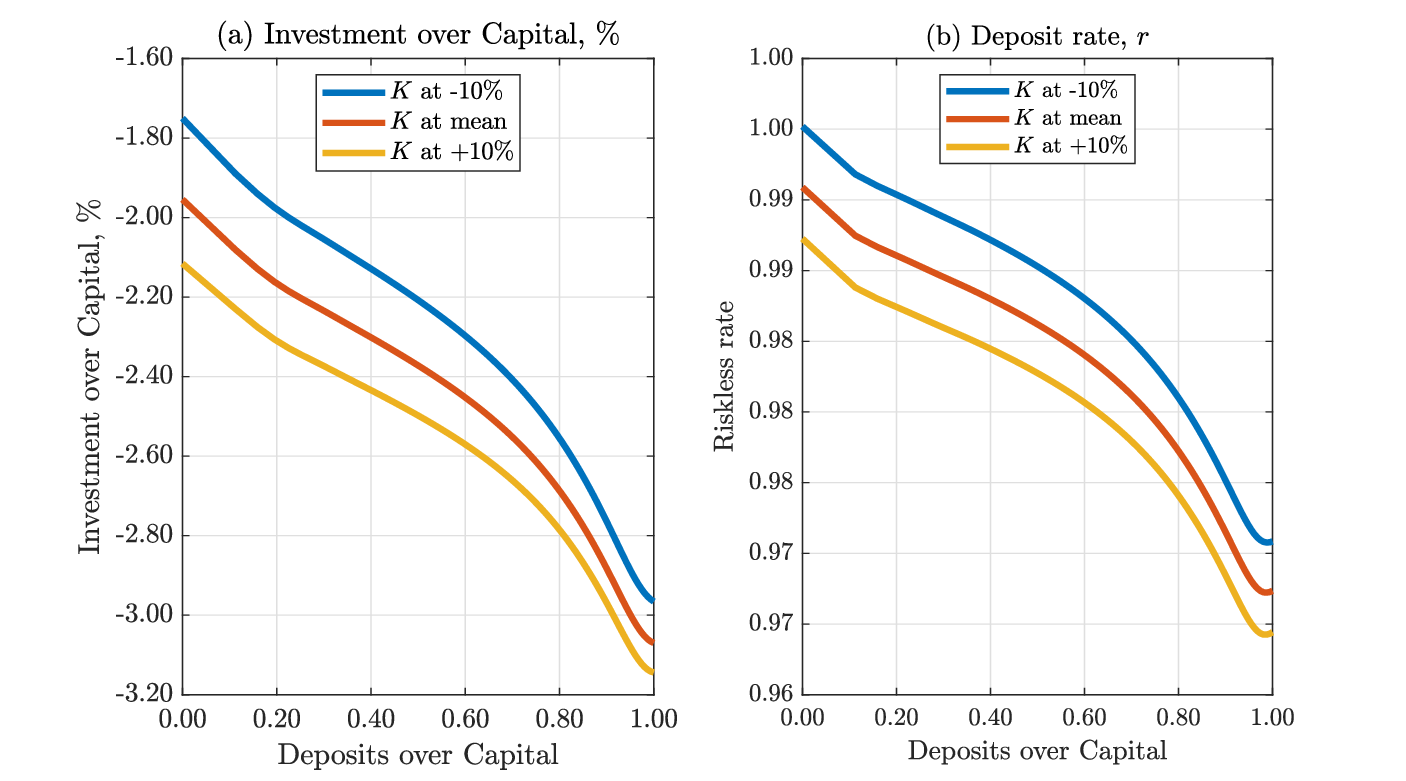}
    \caption{\label{fig:eqbm_policies} Equilibrium objects over the state space. $\zeta = \underline{\zeta}$, $Z=\underline{Z}$. \textit{Mean} refers to the ergodic mean of the competitive equilibrium without taxes.}
\end{center}
\end{figure}

To further illustrate that these patterns are related to experts' capacity to take risk, consider the changes in the \textit{risk premium}. Define risk premium as the expected excess return conditional on the current state
\begin{align}
    RP(\textbf{S}) = \mathbb{E}\left[\frac{\zeta' \left(R(\textbf{S}') + (1 - \delta)q(\textbf{S}') + \Pi (\textbf{S}')\right)}{q(\textbf{S})} \Bigg| \textbf{S}\right]  - r(\textbf{S})
\end{align}
where
\begin{align}
    \Pi(\textbf{S}') = q(\textbf{S}')\Phi\left(\frac{I(\textbf{S}')}{ K'}\right)-\frac{I(\textbf{S}')}{K'}
\end{align}
is the per unit profit from capital good production in state $\textbf{S}'$.

Panel (a) of Figure \ref{fig:risk} shows that the risk premium demanded by the experts for holding capital rises sharply as deposits increase for a given level of capital. Panel (b) of the figure demonstrates the insufficiency of risk sharing between the agents. The vertical axis shows the conditional standard deviation of log consumption in the next period, i.e. $\sigma \left[ \log C_w(\textbf{S}') | \textbf{S}\right]$ and $\sigma \left[ \log C_e(\textbf{S}') | \textbf{S}\right]$.

\begin{figure}[h]
\begin{center}
    \includegraphics[scale = 0.70]{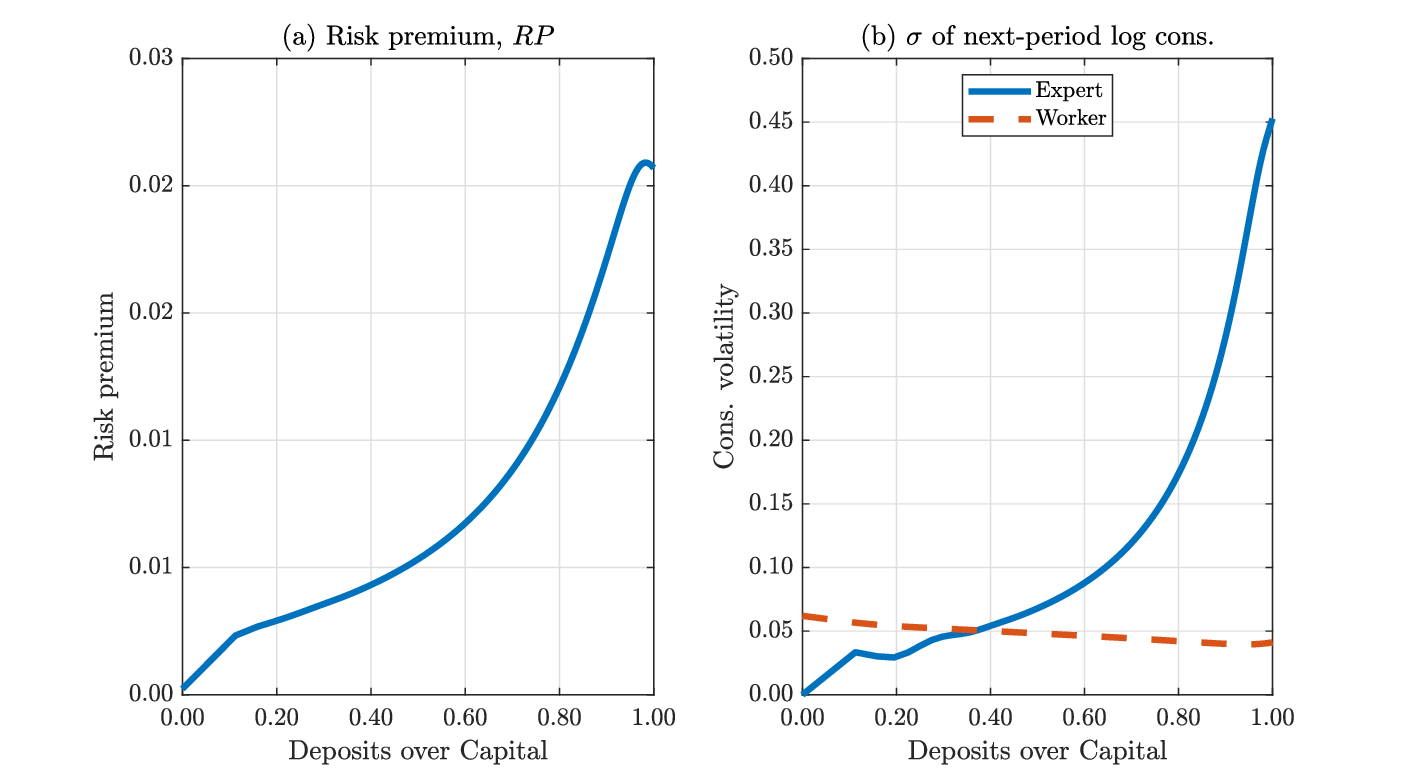}
    \caption{\label{fig:risk} Risk over the state space. $\zeta = \underline{\zeta}$, $Z=\underline{Z}$, $K$ is at mean.}
\end{center}
\end{figure}

The plot shows that the consumption volatility of experts is higher than that of workers when deposits are high and the relationship is reversed for low levels of deposits. Recall that workers are exposed to the aggregate risk through their wages and use deposits to smooth their consumption. We will see that policy attempts to move the wealth of the agents closer to the area where risk sharing is better.

It is important to understand the behavior of workers' consumption volatility when deposits are high. There are two opposing forces. First, as already discussed, more deposits allow the workers to smooth their consumption better and so this force decreases the volatility. Second, when experts have low wealth, not only their consumption but also investment become more sensitive to the exogenous shocks. In this model, this has no direct, contemporaneous effect on the volatility of next-period consumption for workers. This stands in contrast to the models often studied in the literature (e.g. \citet{jermann2012macroeconomic}, \citet{bianchi2018optimal}) where ``working capital'' constraints introduce same-period effects of adverse financing conditions on employment. However, an indirect effect still exists in our dynamic model: more volatile investment in the next period translates into more volatile (from the current-period perspective) output and consumption two periods ahead. Moreover, if a bad shock is realized in the next period, workers will want to reduce their next-period consumption and save more since the shocks are persistent. This effect becomes increasingly strong as we consider states with vanishing expert wealth. Indeed, as panel (b) of Figure \ref{fig:risk} shows, the volatility of workers' consumption starts rising in the right-most region. The same effect can be seen in panel (b) of Figure \ref{fig:eqbm_policies} which shows flattening of the equilibrium riskless rate in the same region. Correspondingly, panel (a) of Figure \ref{fig:risk} also shows the risk premium falling slightly.


\subsection{Financial crises\label{subsec:crises}}

So far, I have not introduced any shock that can clearly be interpreted as one that triggers a financial crisis. Instead, the already assumed aggregate shocks have state-dependent effects on investment and other equilibrium variables. An identical exogenous shock hitting when expert wealth is high causes milder decline in investment then that arising when expert wealth is depleted. This observation motivates the following definition of financial crises in this model.

\begin{dfn}
\label{dfn:crises} A financial crisis is an episode when investment drops by more than two standard deviations below its ergodic mean.
\end{dfn}

Let us now look at the characteristics of such episodes. To do this, I follow an event studies methodology as follows. First, simulate a long series from the model and discard a number of initial periods as burn-in. Second, identify the crisis periods to be the ones where investment is more than two standard deviations below its ergodic mean but above that level in the previous period. Third, take 40 preceding and 20 following periods around each identified crisis period. Then, Figure \ref{fig:crises_eq} plots averages of several variables over all the identified episodes.

\begin{figure}[h]
\begin{center}
    \includegraphics[scale = 0.70]{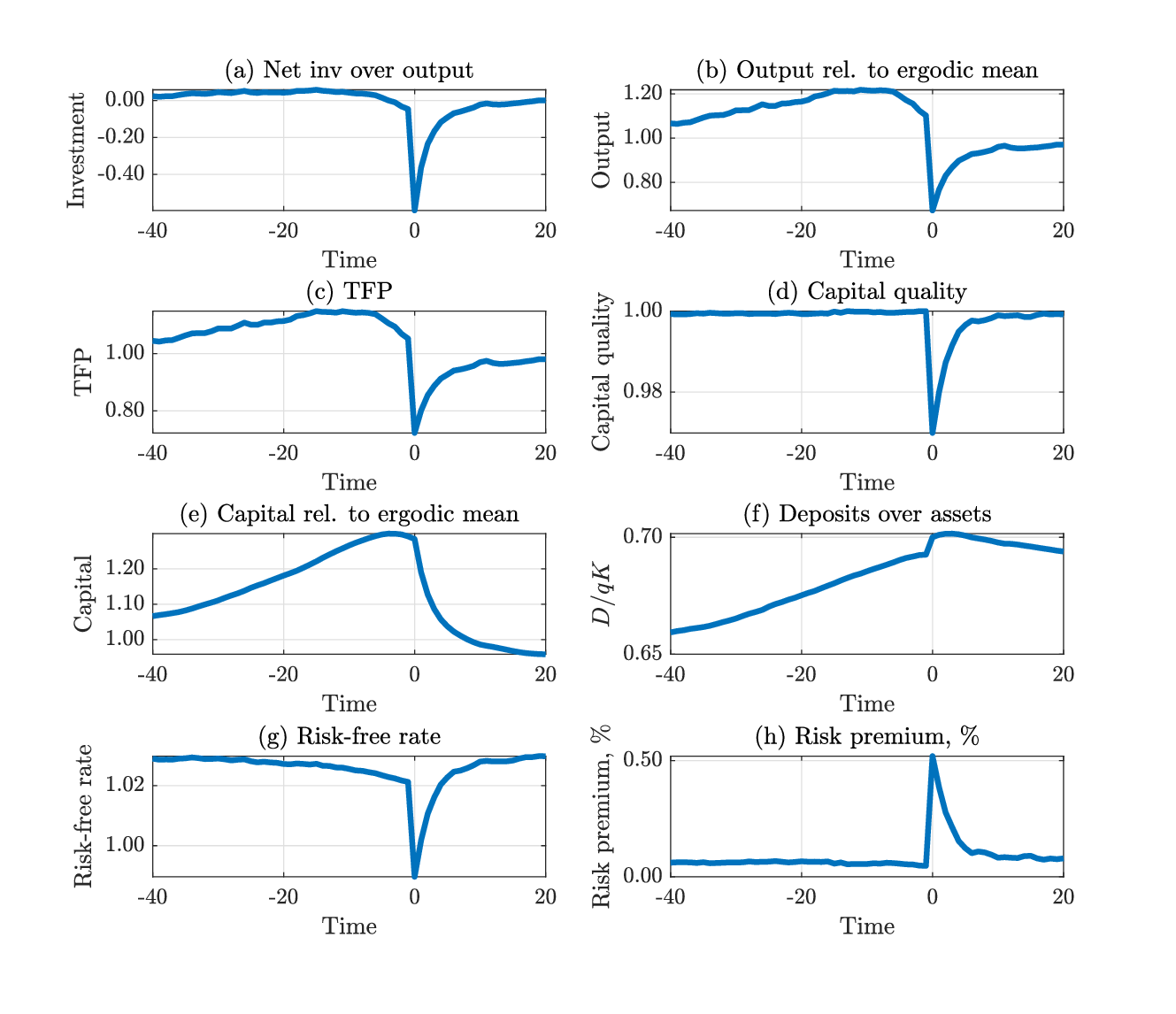}
    \caption{\label{fig:crises_eq} Typical crises defined by a two standard-deviation decline in investment.}
\end{center}
\end{figure}

Crises defined based on large declines in investment happen to be the ones where expert wealth is comparatively low. Panel (f) shows this by slightly modifying the key state variable considered before: the quantity $\frac{D}{qK}$ is more interpretable since it includes the real value of assets in the denominator\footnote{The reason why Figures such as \ref{fig:eqbm_policies} and \ref{fig:risk} plot $\frac{D}{K}$ in the denominator is that the latter is a function of the state variables while $\frac{D}{qK}$ is not guaranteed to be a function.}. What shocks drive the economy into such states? Panels (c) and (d) show that prolonged periods of high productivity combined with a capital quality shock hitting are typically associated with these episodes. Leading to the crises, experts expect future productivity to keep being high and continue to lever up. Panel (a) shows that net investment in the pre-crisis period is persistently positive. In order to lever up, experts issue more deposits. The riskless rate is, again, shaped by supply and demand from the two types. Workers are receiving higher wages and want to save more for the future while expert leverage rising is a force that pushes them to issue less deposits. These forces, combined with rising productivity, lead to increasing deposits and dropping riskless rate.

Once the crisis hits, the impact on experts' wealth is amplified by their leverage. This, in turn, contributes to the unwillingness of the experts to hold capital above and beyond that associated with the drop in productivity. Experts would like to reduce the size of their balance sheets by disinvesting and simultaneously issuing less deposits. As panel (f) shows, the experts do not immediately reduce issued deposits faster than the decrease in assets; as a result, the leverage rises further. Moreover, note that leverage is highly persistent but does, on average, go back to its ergodic mean. The role of risk, as discussed before, can be seen in the dynamics of the risk premium. Panel (h) shows that, on impact, the risk premium jumps to more than $2\%$ annualized.

Finally, Table \ref{tab:calib-targets} reports some statistics over the ergodic distribution. Note the difference between the probability of being in a $\zeta=\underline{\zeta}$ state and the probability of crisis periods as defined by Definition \ref{dfn:crises}. The reason is that when the capital quality shock hits while expert wealth is relatively high, the drop in investment is milder and does not satisfy the crisis definition.

\begin{table}[h]
\begin{center}
    \begin{tabular}{ cc } 
     \toprule
     \textbf{Variable} & \textbf{Value}  \\
     \midrule
     Mean riskless rate & 1.03\\

     Mean deposits over assets & 0.65\\ 
     
     Prob. of being in disaster, $\zeta=\underline{\zeta}$ & 1.45\%\\ 
     
     Prob. of financial crisis & 0.50\%\\ 

     \bottomrule 
    \end{tabular}
    \caption{\label{tab:calib-targets} Characteristics of the equilibrium.}
\end{center}
\end{table}

\section{First-best allocations\label{sec:first-best}}

Before investigating Ramsey-optimal policies, it is instructive to study the first-best allocations. The core friction in the model is market incompleteness. If the two agents are allowed to trade the full set of Arrow-Debreu securities instead of a single riskless claim, the resulting allocation will be efficient by the First Welfare Theorem.

Alternatively, if the planner is allowed to use a lump-sum transfer for workers instead of the distortive tax on labor, efficient allocations will obtain once again. More precisely, suppose the budget constraints \eqref{eq:expert-budget} and \eqref{eq:worker-budget}, respectively, are replaced by the following
\begin{align}
    \quad  c_e +  q(\textbf{S}) k_e' + d_e = &\; \frac{d_e'}{r(\textbf{S})} + R(\textbf{S}) k_e + (1-\delta)k_e +q(\textbf{S})\Phi\left(\frac{i}{ k_e}\right)k_e-i + T(\textbf{S}) \\
    \quad c_w + \frac{d_w'}{r(\textbf{S})} = & \;d_w + W(\textbf{S}) l_w - T(\textbf{S})
\end{align}

This flexibility then allows the planner to attain the social planner's solution where it can freely choose all allocations. The intuition for this is that the lump sum transfers allow the planner to choose consumptions arbitrarily and achieve full risk sharing. Conditional on this, the planner is not interested in distorting either labor or investment choices and, therefore, does not need additional tax instruments. Formally, given an initial state $\textbf{S} = \{K_0, D_0, Z_0, \zeta_0\}$, the planner solves
\begin{align}
    V^{sp}(K_0, D_0, Z_0, \zeta_0) &= \max \mathbb{E}_0 \left[ \sum_{t = 0}^\infty \beta^t \left(\lambda u_e(C_{et}) + (1-\lambda) u_w(C_{wt},L_t) \right) \right] \label{eq:sp-val} \\
    s.t. \;
    K_{t+1} &= \zeta_{t+1} K_t\left(\Phi\left(\frac{I_t}{K_t}\right)  +  (1-\delta)\right)\\
    F(Z_t, K_t, L_t) &= I_t + C_{et} + C_{wt} \label{eq:sp-rc}
\end{align}

Note that not only the initial state $D_0$ but also future levels of deposits become irrelevant for the economy. In fact, an autarkic equilibrium corresponding to the social planner's solution can be constructed where deposits are always at zero and both agents consume their period income augmented by the appropriate lump-sum transfers. Moreover, in this equilibrium, ratios of marginal utilities of the agents are equalized across states and risk is perfectly shared
\begin{align}
    \frac{u_{ec}(C_{wt+1})}{u_{ec}(C_{wt})} = \frac{u_{wc}(C_{wt+1},L_{t+1})}{u_{wc}(C_{wt},L_t)} \label{eq:sp-euler}
\end{align}

The Pareto weight $\lambda$ affects the allocations chosen by the social planner and takes the role of the initial state $D_0$. That is, $\lambda$ corresponds to a particular initial wealth distribution. More formally, for any given initial state, a $\lambda$ can be picked so that the allocations chosen by the social planner coincide with the equilibrium arising in the model augmented with Arrow-Debreu securities. Finally, note that because our parametrization allows for an income effect on labor supply, the Pareto weight affects not only the allocation of consumption but also labor and investment levels. In particular, as $\lambda$ increases, the planner allocates lower consumption to the worker which must correspond to higher equilibrium labor, as well as higher investment and capital levels. We will further explore the properties of the first-best allocations in Section \ref{sec:ramsey}.

\section{Optimal tax policies\label{sec:ramsey}}

In this section, I formally state the Ramsey planner's problem, show how to exploit its recursive structure and discuss the theoretical and numerical properties of the optimal plans. Importantly, I will assume that the planner operates under \textit{full commitment}.

\subsection{Sequence formulation \label{subsec:ramsey-seq}}

Based on the objective defined in Section \ref{subsec:ramseydef}, the Ramsey planner chooses sequences of taxes $\{\tau_{lt}\}$, $\{\tau_{dt}\}$ and $\{\tau_{kt}\}$ in order to maximize the objective subject to the equilibrium conditions of the model. Appendix \ref{app:eqbm-cond} shows that the problem reduces to
\begin{align}
    V^{pl}(K_0, D_0, Z_0, \zeta_0) &= \max \mathbb{E}_0 \left[ \sum_{t = 0}^\infty \beta^t \left(\lambda u_e(C_{et}) + (1-\lambda) u_w(C_{wt},L_t) \right) \right] \label{eq:val-plan}\\
    s.t. \;-\frac{u_{wl}(C_{wt},L_t)}{u_{wc}(C_{wt},L_t)} &= (1 - \tau_{lt})F_L(Z_t, K_t, L_t) \label{eq:mrs} \\
    u_{wc}(C_{wt},L_t) &= r_t \beta \mathbb{E}_t\left[ u_{wc}(C_{wt+1},L_{t+1}) \right]\\
    (1-\tau_{dt})u_{ec}(C_{et}) &= r_t \beta \mathbb{E}_t\left[ u_{ec}(C_{et+1}) \right]\label{eq:euler-expert} \\
    \mathbb{E}_t\left[u_{ec}(C_{et+1})\right]\frac{1 + \tau_{kt}}{1-\tau_{dt}}r_t &= \mathbb{E}_t\left[u_{ec}(C_{et+1}) \left(\frac{\zeta_{t+1} \left(R_{t+1} + (1 - \delta)q_{t+1} + \Pi_{t+1}\right)}{q_t} \right)\right] \label{eq:arbitrage}\\
    1 &= q_t\Phi'\left(\frac{Y_t - C_{et} + C_{wt}}{K_t}\right) \label{eq:cap-price}\\
    K_{t+1} &= \zeta_{t+1} K_t\left(\Phi\left(\frac{Y_t - C_{et} + C_{wt}}{K_t}\right)  +  (1-\delta)\right)\\
    \frac{D_{t+1}}{r_t} + C_{wt} &= D_t + F_L(Z_t,K_t,L_t) L_t (1-\tau_{lt}) \label{eq:dlom}\\
    D_{t+1} &\in [\underline{D}_{t+1}, \overline{D}_{t+1}] \label{eq:ram-bounds}
\end{align}

where it is assumed that deposits satisfy natural borrowing limits and the maximum is taken over the sequences $\{\tau_{lt}\},\{\tau_{dt}\},\{\tau_{kt}\},\{\tau_{kt}\}$ as well as all equilibrium quantities and prices. Observe that \eqref{eq:dlom} is the law of motion of workers' deposits and, therefore, does not contain the transfer $T_t$\footnote{The corresponding law of motion for experts is dropped by an application of Walras law.}.

As usual, it is convenient to rewrite the problem into the primal form so that the planner chooses sequences of allocations subject to technological and implementability constraints. In doing so, several conditions from above can be eliminated. First, we can express $\tau_{lt}$ from condition \eqref{eq:mrs} and substitute into condition \eqref{eq:dlom}. Then, \eqref{eq:mrs} can be dropped because the labor tax can always be set residually to satisfy the condition. Similarly, the availability of $\tau_{dt}$ and $\tau_{kt}$ allows us to get rid of conditions \eqref{eq:euler-expert} and \eqref{eq:arbitrage}. Moreover, the constraint \eqref{eq:cap-price} is only a pricing equation for capital and can be dropped since $q_t$ does not appear anywhere else after eliminating \eqref{eq:arbitrage}.

The primal form of the problem can then be written as follows
\begin{align}
    V^{pl}(K_0, D_0, Z_0, \zeta_0) &= \max \mathbb{E}_0 \left[ \sum_{t = 0}^\infty \beta^t \left(\lambda u_e(C_{et}) + (1-\lambda) u_w(C_{wt},L_t) \right) \right] \label{eq:ram-val} \\
    s.t. \;\beta \mathbb{E}_t[u_{wc}(C_{wt+1}, L_{t+1})]D_{t+1} &= D_t u_{wc}(C_{wt}, L_t) -u_{wc}(C_{wt}, L_t) C_{wt} - u_{wl}(C_{wt}, L_t) L_t \label{eq:impl} \\
    K_{t+1} &= \zeta_{t+1} K_t\left(\Phi\left(\frac{I_t}{K_t}\right)  +  (1-\delta)\right)\\
    F(Z_t, K_t, L_t) &= I_t + C_{et} + C_{wt}\\
    D_{t+1} &\in [\underline{D}_{t+1}, \overline{D}_{t+1}] \label{eq:ram-primal-bounds}
\end{align}

where the maximum is now taken only over the allocation sequences and taxes do not appear in the problem.

Constraint \eqref{eq:impl} is the implementability constraint that commonly appears in the literature on optimal taxation with incomplete markets (e.g. \citet{aiyagari2002optimal}). Also called the \textit{measurability} constraint, this equation reflects the riskless nature of deposits. The left-hand side of the condition is an expectation over the shocks to be realized in period $t+1$ while $D_{t+1}$ is only measurable with respect to the information available at time $t$. Consequently, the planner's choices of worker consumption and labor across shock realizations for period $t+1$ have to all satisfy a common restriction.

\subsection{Recursive formulation \label{subsec:ramsey-rec}}

A recursive formulation is desirable in order to solve the planner's problem numerically. However, the constraint \eqref{eq:impl} precludes the transition to a recursive form using standard methods. The difficulty arises due to both $t$ and $t+1$ choice variables being in the equation. Such \textit{forward-looking} constraints often occur when studying optimal contracts and Ramsey-optimal policies. Several techniques have been developed to deal with such structures. Here, I will use the Lagrangian approach from \citet{marcet2019recursive}. 

Letting $\eta_t$ be the Lagrange multiplier associated with \eqref{eq:impl}, the following recursive saddle-point problem becomes the key object
\begin{align}
    \widetilde{V}^{pl}(K, D, Z, \zeta, \mu) &= \inf_{\eta} \sup_{C_e,C_w,L,I,D'} \lambda u_e(C_e) + (1-\lambda) u_w(C_w,L) \label{eq:saddle-begin} \\&- \mu D u_{wc}(C_w,L) +\eta \left(D u_{wc}(C_w,L) - u_{wc}(C_w,L)C_w - u_{wl}(C_w,L)L \right)  \notag  \\& + \beta \mathbb{E}\left[\widetilde{V}^{pl}(K', D', Z', \eta) \right] \notag \\
    s.t. \;
    K' &= \zeta' K\left(\Phi\left(\frac{I}{K}\right)  +  (1-\delta)\right)\\
    F(Z, K, L) &= I + C_{e} + C_{w}\\
    D' &\in [\underline{D}(\widetilde{\textbf{S}}), \overline{D}(\widetilde{\textbf{S}})] \label{eq:saddle-end}
\end{align}

where $\widetilde{\textbf{S}}=\{K, D, Z, \zeta, \mu\}$ is the original state vector of the economy augmented by a Lagrange multiplier.

Under some conditions outlined in \citet{marcet2019recursive}, the value in the original problem is given by
\begin{align}
    {V}^{pl}(K_0, D_0, Z_0, \zeta_0) = \widetilde{V}^{pl}(K_0, D_0, Z_0, \zeta_0, 0)
\end{align}

and optimal allocations are found similarly. Moreover, the corresponding optimal tax policies can be backed out by plugging in the optimal allocations into the corresponding equilibrium conditions. That procedure yields
\begin{align}
    \tau_{k}\big(\widetilde{\textbf{S}}\big) &= -1 + \mathbb{E}\left[\beta \frac{u_{ec}\left(C_{e}\big(\widetilde{\textbf{S}}'\big)\right)}{u_{ec}\left(C_{e}\big(\widetilde{\textbf{S}}\big)\right)} \left(\frac{\zeta' \left(R\big(\widetilde{\textbf{S}}\big) + (1 - \delta)q\big(\widetilde{\textbf{S}}'\big) + \Pi\big(\widetilde{\textbf{S}}'\big)\right)}{q\big(\widetilde{\textbf{S}}\big)} \right)\right] \label{eq:cap-tax}\\
    \tau_l\big(\widetilde{\textbf{S}}\big) &= 1+\frac{u_{wl}\left(C_w\big(\widetilde{\textbf{S}}\big),L\big(\widetilde{\textbf{S}}\big)\right)}{u_{wc}\left(C_w\big(\widetilde{\textbf{S}}\big),L\big(\widetilde{\textbf{S}}\big)\right)F_L\left(Z, K, L\big(\widetilde{\textbf{S}}\big)\right)} \label{eq:lab-tax}\\
    \tau_d\big(\widetilde{\textbf{S}}\big) &= 1 - \frac{\mathbb{E}\left[ u_{ec}\left(C_e\big(\widetilde{\textbf{S}}'\big)\right) \right]}{u_{ec}\left(C_e\big(\widetilde{\textbf{S}}\big)\right)} \bigg/ \frac{\mathbb{E}\left[ u_{ec}\left(C_w\big(\widetilde{\textbf{S}}'\big),L\big(\widetilde{\textbf{S}}'\big)\right) \right]}{u_{ec}\left(C_w\big(\widetilde{\textbf{S}}\big),L\big(\widetilde{\textbf{S}}\big)\right)}\label{eq:dep-tax}
\end{align}

where $R\big(\widetilde{\textbf{S}}\big)$, $\Pi\big(\widetilde{\textbf{S}}\big)$ and $q\big(\widetilde{\textbf{S}}\big)$ are functions of allocations defined as before.

The recursive saddle-point problem above is simply a rewritten form of the Lagrangian associated with the original problem. Splitting the condition \eqref{eq:impl} into current period (left-hand side) and next period (right-hand side) components is the central insight for this transition. The multiplier $\eta$ becomes a state variable with the state $\mu$ representing the multiplier on the implementability constraint associated with the previous period. It is passed along to the next period and reflects promises made by the planner in the previous periods.

Recall that the forward-looking constraint facing the planner is essentially the Euler equation of the worker type. When choosing the state-dependent consumption for workers between any two periods, the planner must make sure to satisfy this optimality condition. Intuitively, the planner ``convinces'' the workers to choose the desirable level of consumption in the current period by ``promising'' certain outcomes for the next period. The extent to which the Lagrange multiplier differs from zero is related to the extent of \textit{time-inconsistency} in the optimal plans. Indeed, whenever the planner can benefit from ignoring previous commitments and re-optimizing given the current state of the economy, we must have
\begin{align}
    \widetilde{V}^{pl}(K, D, Z, \zeta, 0) > \widetilde{V}^{pl}(K, D, Z, \zeta, \mu)
\end{align}

Notice again that the state vector that can represent the economy under optimal policy, $\widetilde{\textbf{S}}$, is richer than the one that represents the economy under no tax policies, $\textbf{S}$. Accordingly, the equilibria induced by Ramsey-optimal tax sequences do not generally satisfy the recursive equilibrium Definition \ref{dfn:eqbm} which presumes tax policies only depend on $\textbf{S}$. The following definition formalizes the notion of recursive competitive equilibria under Ramsey policies.

\begin{dfn}
Let $\tau_k^*\big(\widetilde{\textbf{S}}\big)$, $\tau_l^*\big(\widetilde{\textbf{S}}\big)$ and $\tau_d^*\big(\widetilde{\textbf{S}}\big)$ be the tax policies and $\eta^*\big(\widetilde{\textbf{S}}\big)$ be the Lagrange multiplier that solve the planner's problem. The corresponding recursive competitive equilibrium is a collection of value, policy and price functions of the augmented state $\widetilde{\textbf{S}}$, as well as a law of motion for the state, $\bm{\widetilde{\Gamma}}$, such that
    \begin{enumerate}
        \item Policy functions solve household utility maximization problems given prices, the laws of motion and continuation values.
        \item Value functions solve the corresponding functional equations.
        \item Firms maximize their profits given prices.
        \item Aggregation holds and markets clear.
        \item The law of motion of the aggregate state induced by the policy functions, the evolution of exogenous shocks and the Lagrange multiplier function $\eta^*\big(\widetilde{\textbf{S}}\big)$ is consistent with that used by the households.
    \end{enumerate}
\end{dfn}

\subsection{Characterization \label{subsec:ramsey-char}}

This subsection characterizes the optimal plans -- both analytically and numerically -- and compares the allocations obtained by the planner to the competitive equilibrium as well as first-best allocations. I then turn to comparisons of crisis dynamics.

\subsubsection{An analytical result \label{subsubsec:ramsey-thm}}

I begin by stating a crucial analytical result which characterizes the optimal tax on capital purchases.

\begin{thm}
\label{thm:cap-tax}
For any initial state $\{K_0, D_0, Z_0, \zeta_0\}$, the optimal tax on capital purchases, $\tau_{kt}$ is zero for all $t\geq 0$ and for all paths of shock realizations.
\end{thm}
\begin{proof}
See Appendix \ref{app:cap-tax-proof}.
\end{proof}

This result is predicated on the labor and deposit issuance taxes being set optimally. The planner would like to equalize the marginal utilities of consumption (with Pareto weights applied) between workers and experts. How can the tax on capital purchases help achieve this objective? Recall that, as stated by the budget constraint of the experts, \eqref{eq:expert-budget}, and the government, \eqref{eq:gov-budget}, tax revenues are transferred to the experts. Hence, the tax on capital purchases cannot redistribute resources directly. However, the tax affects two key equilibrium objects that do redistribute -- the wage $W$ and the riskless rate $r$. It affects the former dynamically by altering investment\footnote{In fact, the capital purchases tax allows the planner to control investment arbitrarily albeit not without consequences.} while the latter changes -- both contemporaneously and dynamically -- due to the arbitrage condition \eqref{eq:arbitrage}. The essential difference between the two channels is that whereas the capital purchases tax affects wages and riskless rates indirectly, the taxes on labor and deposit issuance, respectively, affect prices directly. 

Consider the original problem \eqref{eq:mrs}-\eqref{eq:ram-bounds}. Letting the planner maximize over the labor tax and the wage is equivalent to letting the planner maximize only over the after-tax wage\footnote{In the system \eqref{eq:mrs}-\eqref{eq:ram-bounds}, the wage has already been substituted out.}. Similarly, one can substitute out the riskless rate in the system of equations so that letting the planner choose the tax rate is equivalent to letting it choose the after tax riskless rate faced by the experts. The tax on capital purchases affects wages faced by workers and riskless rates faced by experts precisely the same way. Using the capital tax, however, achieves the changes by distorting investment. It then follows by an application of the envelope theorem that as long as the labor and deposit issuance taxes are set optimally, the tax on capital purchases has to be zero. In effect, all gains from manipulating the wage and the riskless rate are attained by the direct taxes and there remains no additional motive for altering investment. Appendix \ref{app:cap-tax-proof} gives a simple formal proof using the primal approach.

\subsubsection{Optimal allocations and taxes \label{subsubsec:optimal-tax}}

I solve the planner's recursive problem using a value function iteration method the details of which are in Appendix \ref{app:numerical-method}. Before turning to optimal tax schedules, let us inspect the optimal allocations chosen by the Ramsey planner and compare them to the competitive equilibrium without taxes as well as the first-best. When presenting optimal allocations and taxes, a value for the Lagrange multiplier state $\mu$ should be fixed. As discussed before, $\mu = 0$ is special in that it maps to time-0 choices of the planner.\footnote{The mapping is not one-to-one. States with $\mu=0$ may be visited also when $t>0$ with possibly infinite occurrences.} The figures that follow below plot allocations and taxes over slices of the state space with $\mu=0$ as a natural benchmark and a comparison point to the competitive equilibrium without policy and the first-best. However, the planner's redistributive motives strongly influence these choices and make them dependent on the Pareto weight parameter $\lambda$. Indeed, the redistributive motive would lead the planner to choose broadly similar time-0 policies even in the absence of any exogenous shocks in the model.\footnote{In the absence of exogenous shocks, the single tradable one-period claim completes the markets. Therefore, the competitive equilibrium without policies is efficient and the planner only acts to redistribute.} To explore the role of optimal policies in improving risk sharing, we will look at crisis dynamics under optimal policies and compare to those presented in Section \ref{subsec:crises}.

\begin{figure}[h]
\begin{center}
    \includegraphics[scale = 0.70]{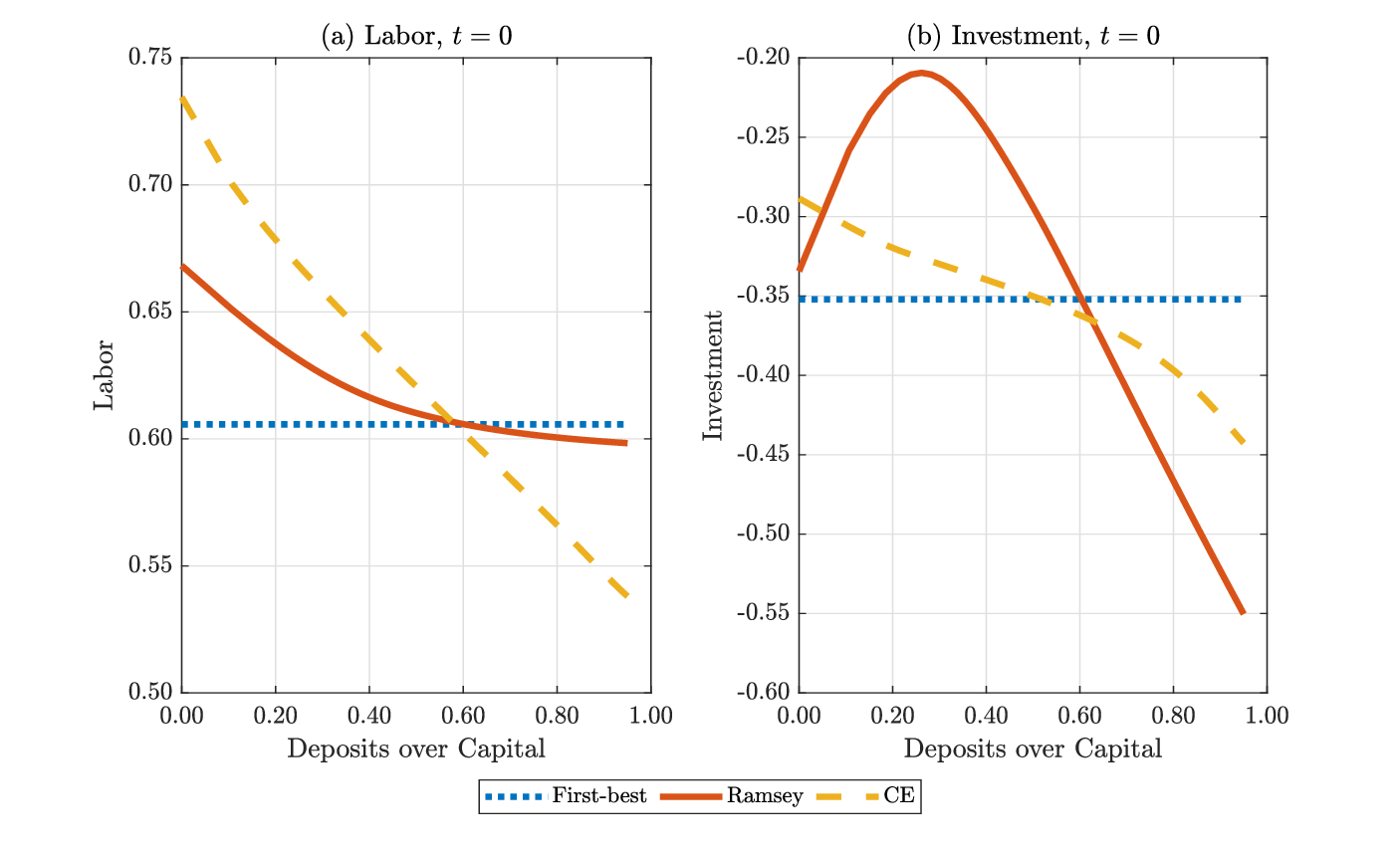}
    \caption{\label{fig:il_ramsey} Time-0 labor and investment under the planner's optimal policies. $K$ is fixed at the mean of the competitive equilibrium without policies, $\zeta=\underline{\zeta}$ and $Z=\underline{Z}$.}
\end{center}
\end{figure}

Figure \ref{fig:il_ramsey} shows the optimal labor and investment values at time 0 under the planner's optimal policies. Recall from Section \ref{sec:first-best} that first-best policies should not vary across the deposits dimension. In contrast, labor in the competitive equilibrium under no policies is a decreasing function of deposits due to the wealth effect. The Ramsey-optimal policy strives to increase labor supply when deposits are high and to decrease labor when deposits are low. At the same time, optimal choices of labor are dictated not only by efficiency but also redistribution concerns. The planner potentially uses the tax on labor to transfer resources to or from the experts. More generally, there are two ways for the planner to move closer to the first-best allocations. First, the planner can try to induce desirable outcomes at any given state of the economy. Second, the planner can alter the outcomes in the current period so as to move the state of the economy to where it is easier to induce allocations close to the first-best.

Panel (b) of Figure \ref{fig:il_ramsey} indicates that the optimal schedule of investment is more heavily affected by the redistributive motive because in much of the state space, the investment chosen by the planner is further away from the first-best level than the investment under no policies. We will soon confirm the importance of the deposit issuance tax -- key in shaping investment -- in redistribution.

\begin{figure}[h]
\begin{center}
    \includegraphics[scale = 0.70]{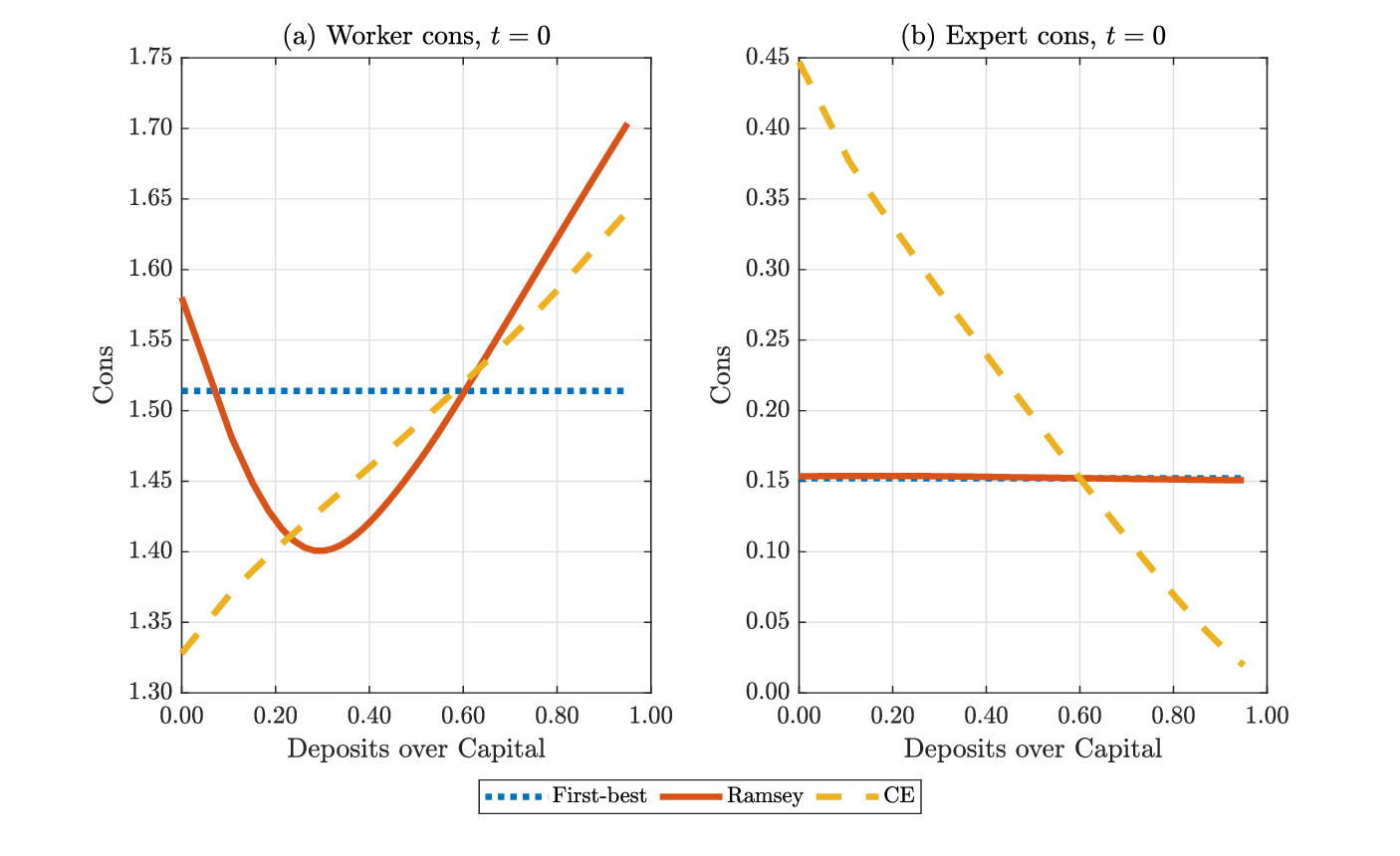}
    \caption{\label{fig:cons_ramsey} Time-0 consumptions under the planner's optimal policies. $K$ is fixed at the mean of the competitive equilibrium without policies, $\zeta=\underline{\zeta}$ and $Z=\underline{Z}$.}
\end{center}
\end{figure}

Figure \ref{fig:cons_ramsey} illustrates the impact of optimal policies on time-0 consumption of the agents. Panel (b) shows that experts' consumption varies greatly over the states where the economy spends much of the time (recall Table \ref{tab:calib-targets}). The Ramsey-optimal allocations feature substantially less variation in experts' consumption and, in that regard, get very close to the first-best outcomes. Panel (a), on the other hand, indicates that workers' consumption under optimal policies is farther away from the first-best than the allocation that obtains in the absence of policies. Again, as we will see, dynamic considerations shape such apparent paradoxes.

\begin{figure}[h]
\begin{center}
    \includegraphics[scale = 0.70]{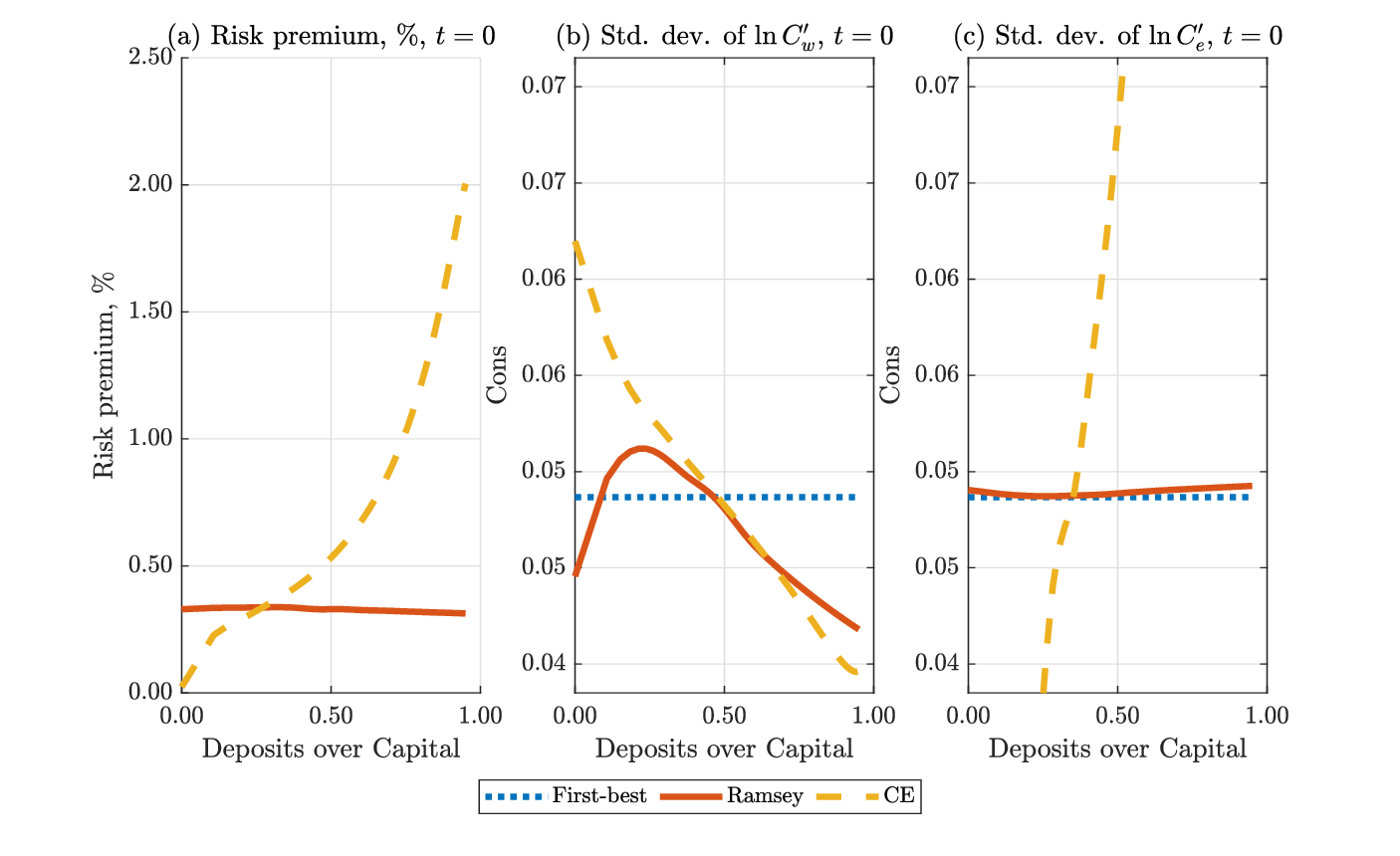}
    \caption{\label{fig:risk_ramsey} Time-0 risk premia and conditional volatility of consumption under the planner's optimal policies. $K$ is fixed at the mean of the competitive equilibrium without policies, $\zeta=\underline{\zeta}$ and $Z=\underline{Z}$.}
\end{center}
\end{figure}

Turning towards dynamic effects, Figure \ref{fig:risk_ramsey} mirrors Figure \ref{fig:risk} and explores the behavior of the risk premium and one-period ahead consumption volatility under the optimal plan.\footnote{The risk premium shown in the figure is now defined in consideration of policy, i.e. $RP = \mathbb{E}\left[\frac{\zeta' \left(R(\widetilde{\textbf{S}}') + (1 - \delta)q(\widetilde{\textbf{S}}') + \Pi (\widetilde{\textbf{S}}')\right)}{q(\widetilde{\textbf{S}})} \Bigg| \widetilde{\textbf{S}}\right]  - \frac{r(\widetilde{\textbf{S}})}{1-\tau_d(\widetilde{\textbf{S}})}$} Panel (a) displays a significant reduction in the risk premium over much of the state space. Panels (b) and (c) show that the planner goes a long way in equalizing the consumption volatilities of the agents. In the right-most region of the state space where deposits are relatively high, experts' consumption volatility is reduced at the expense of that of workers. The opposite pattern holds when deposits are low. In the intermediate region, however, the planner can reduce the volatility of experts' consumption without increasing that of workers' consumption.

I now turn to the taxes and transfers that induce these allocation patterns. Optimal taxes can be backed out using formulas \eqref{eq:cap-tax}-\eqref{eq:dep-tax} and transfers can be computed from the government's budget constraint. Panel (a) of Figure \ref{fig:transf} depicts the transfer $T(\widetilde{\textbf{S}})$ as a fraction of output. As expected, the transfer is positive to the right of a threshold and negative to the left of it. Transfers are very large compared to output which reflects the strength of the planner's redistributive motive at time 0. Panel (b) speaks to the cyclicality of the transfer policy by plotting the difference between transfer-to-output ratios for the disaster state and the normal state where productivity is at its median. Unsurprisingly, the magnitude of transfers is amplified when a bad shock is realized. Lastly, the reason for transfers being U-shaped and decreasing in absolute value as deposits decline, is related to the composition of the transfer.

\begin{figure}[h]
\begin{center}
    \includegraphics[scale = 0.70]{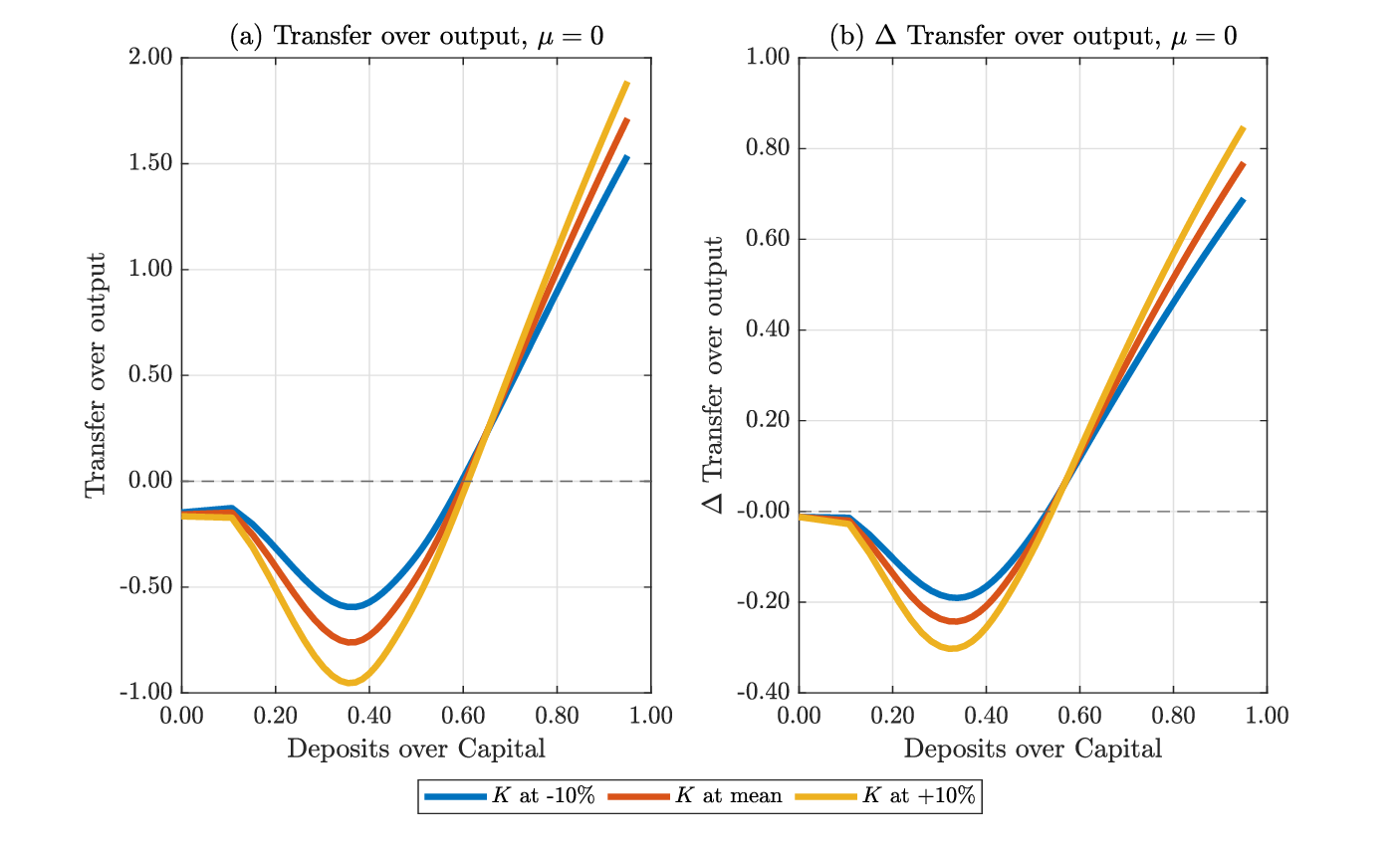}
    \caption{\label{fig:transf} Implied time-0 transfers. $K$ is fixed at the mean of the competitive equilibrium without policies, $\zeta=\underline{\zeta}$ and $Z=\underline{Z}$. Panel (b) shows the difference between $\zeta=\underline{\zeta}$, $Z=\underline{Z}$ and $\zeta=1$, $Z$ at median.}
\end{center}
\end{figure}

The planner can use either of the two tax instruments to redistribute but the relative efficiency of these tools varies over the state space. Figures \ref{fig:tax_lab} and \ref{fig:tax_r} show the optimal time-0 taxes on labor and deposit issuance. Comparing these plots to Figure \ref{fig:transf} reveals that when deposits are relatively high, positive transfers to experts are financed exclusively using the deposit issuance tax. The positive deposit issuance tax reduces experts' willingness to issue deposits and creates downward pressure on the riskless rate faced by the workers. Because the revenues from the tax are rebated to experts they gain resources and, at the same time, perceive deposit issuance to be more costly.

\begin{figure}[h]
\begin{center}
    \includegraphics[scale = 0.70]{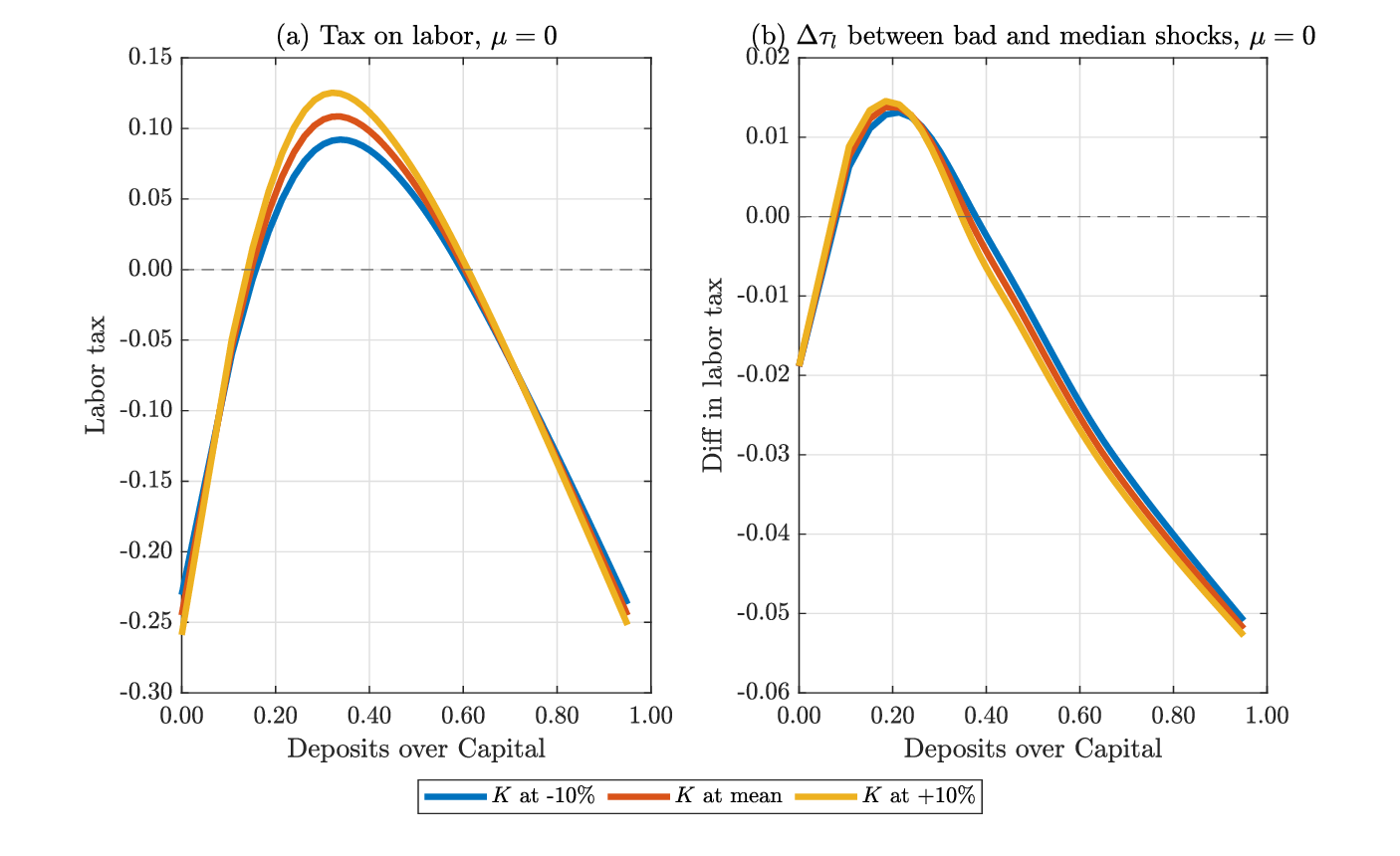}
    \caption{\label{fig:tax_lab} Implied time-0 labor tax. $K$ is fixed at the mean of the competitive equilibrium without policies, $\zeta=\underline{\zeta}$ and $Z=\underline{Z}$. Panel (b) shows the difference between $\zeta=\underline{\zeta}$, $Z=\underline{Z}$ and $\zeta=1$, $Z$ at median.}
\end{center}
\end{figure}

Rents of this nature are similar to rents arising from monopoly in the deposit market. Suppose for a moment that only a single expert exists who internalizes the dependence of the equilibrium rate on the quantity of deposits it chooses. Then, an agent with such market power would choose to issue less deposits and would pay a lower equilibrium rate. As noted before, redistribution through the deposit issuance tax can also be thought of through the lens of quantity controls. If the policy-maker imposes a binding cap on deposits that experts can issue,\footnote{An exact cap corresponding to the optimal tax can be trivially set to the equilibrium quantity of deposits under optimal policy. A quantity control corresponding to a negative tax would instead be a quantity floor.} the equilibrium rate will be bid down creating rents for the experts.

Raising revenue using the tax on deposits leads to investment deviating further away from the first-best, as evident in panel (b) of Figure \ref{fig:il_ramsey}. It turns out, however, that this tax allows the planner to quickly move the economy to a region of the state space where close-to-efficient investment can be achieved. To the extent that the deposit tax takes care of redistribution when deposits are high, the labor tax is left to its role of improving the efficiency of labor supply. Consistent with panel (a) of Figure \ref{fig:il_ramsey}, the planner subsidizes labor and increases equilibrium outcome towards the first-best level.

\begin{figure}[h]
\begin{center}
    \includegraphics[scale = 0.70]{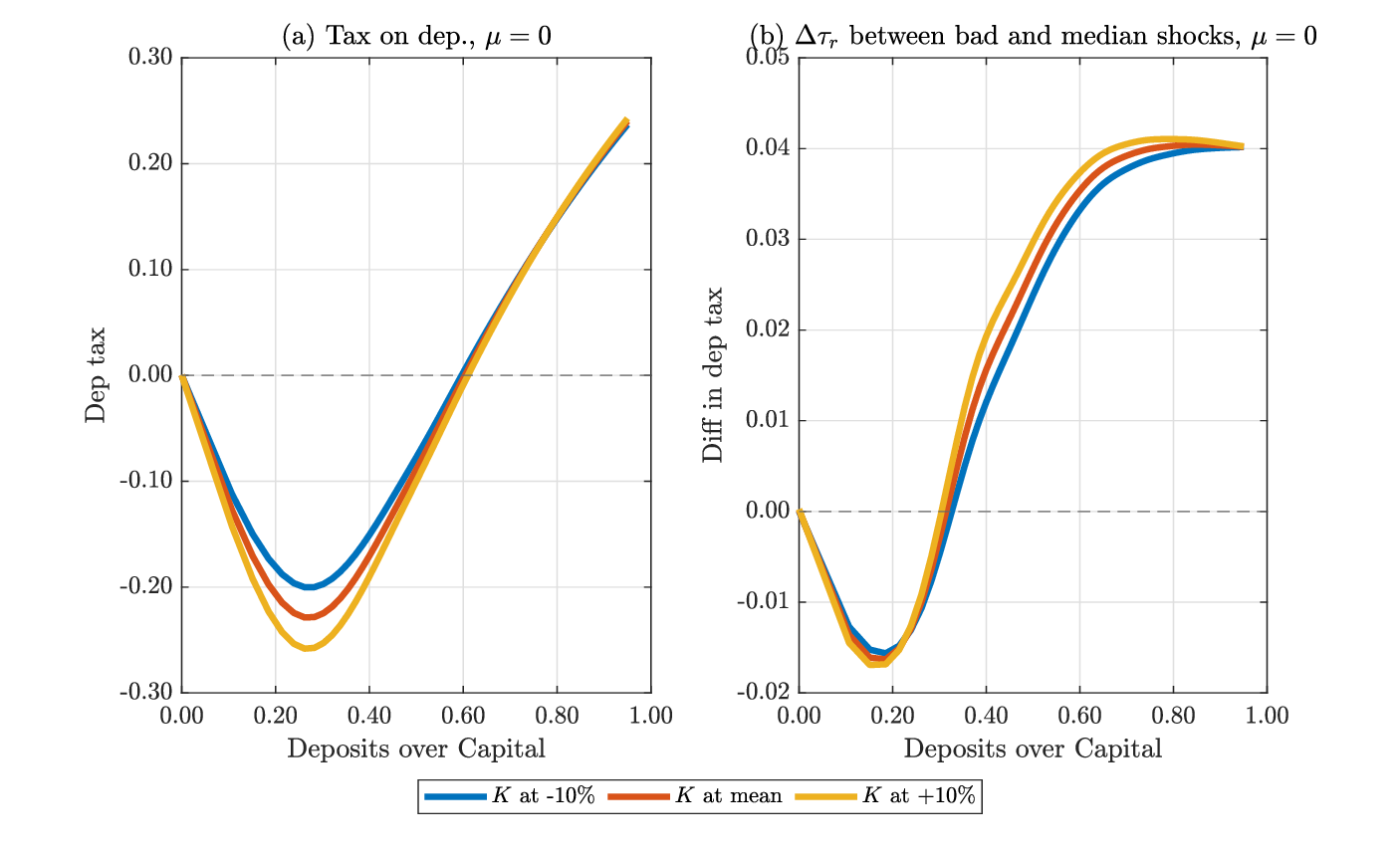}
    \caption{\label{fig:tax_r} Implied time-0 tax on deposit issuance. $K$ is fixed at the mean of the competitive equilibrium without policies, $\zeta=\underline{\zeta}$ and $Z=\underline{Z}$. Panel (b) shows the difference between $\zeta=\underline{\zeta}$, $Z=\underline{Z}$ and $\zeta=1$, $Z$ at median.}
\end{center}
\end{figure}

As we move to the left along the horizontal axis, a key force shaping changes in taxes is the efficiency of using the tax on deposit issuance in redistribution. As the tax base declines, a larger distortion is needed to redistribute. Consequently, the labor tax begins to take the role of the redistributing device. As we get closer to the region with very low deposits, the labor tax becomes the main driver of transfers and, as such, needs to be negative in order to channel resources towards the comparatively poorer workers. This might seem surprising given that equilibrium labor under optimal policies is lower than that under no policies, as panel (a) of Figure \ref{fig:il_ramsey} shows. However, observe that the policies not only change the marginal incentives but also create an income effect on labor supply. Moreover, the income effects induced by transfers to workers with nearly-zero wealth are relatively large.

This section concludes with a result that sheds further light on the dependence of optimal taxes on the efficiency of using individual tax instruments.

\begin{thm}
\label{thm:inel-labor}
When labor supply is perfectly inelastic, i.e. $\nu = \infty$, the Ramsey planner can achieve the first-best outcomes. Taxes can be raised by only using the labor tax and setting the taxes on deposit issuance and capital purchases to zero. Moreover, the planner's solution is time-consistent.
\end{thm}
\begin{proof}
See Appendix \ref{app:inel-labor-proof}.
\end{proof}

The intuition for the result is simple: the labor tax amounts to a lump-sum tax when labor supply is perfectly inelastic. Since the transfer is already assumed to be lump-sum from experts' viewpoint, the planner achieves the first-best as in Section \ref{sec:first-best}. Recall, that in the first-best risk, is fully shared and there is no reason for the planner to use the tax on deposit issuance.

\subsubsection{Crisis dynamics under optimal policy \label{subsubsec:ramsey-crises}}

\begin{figure}[h!]
\begin{center}
    \includegraphics[scale = 0.70]{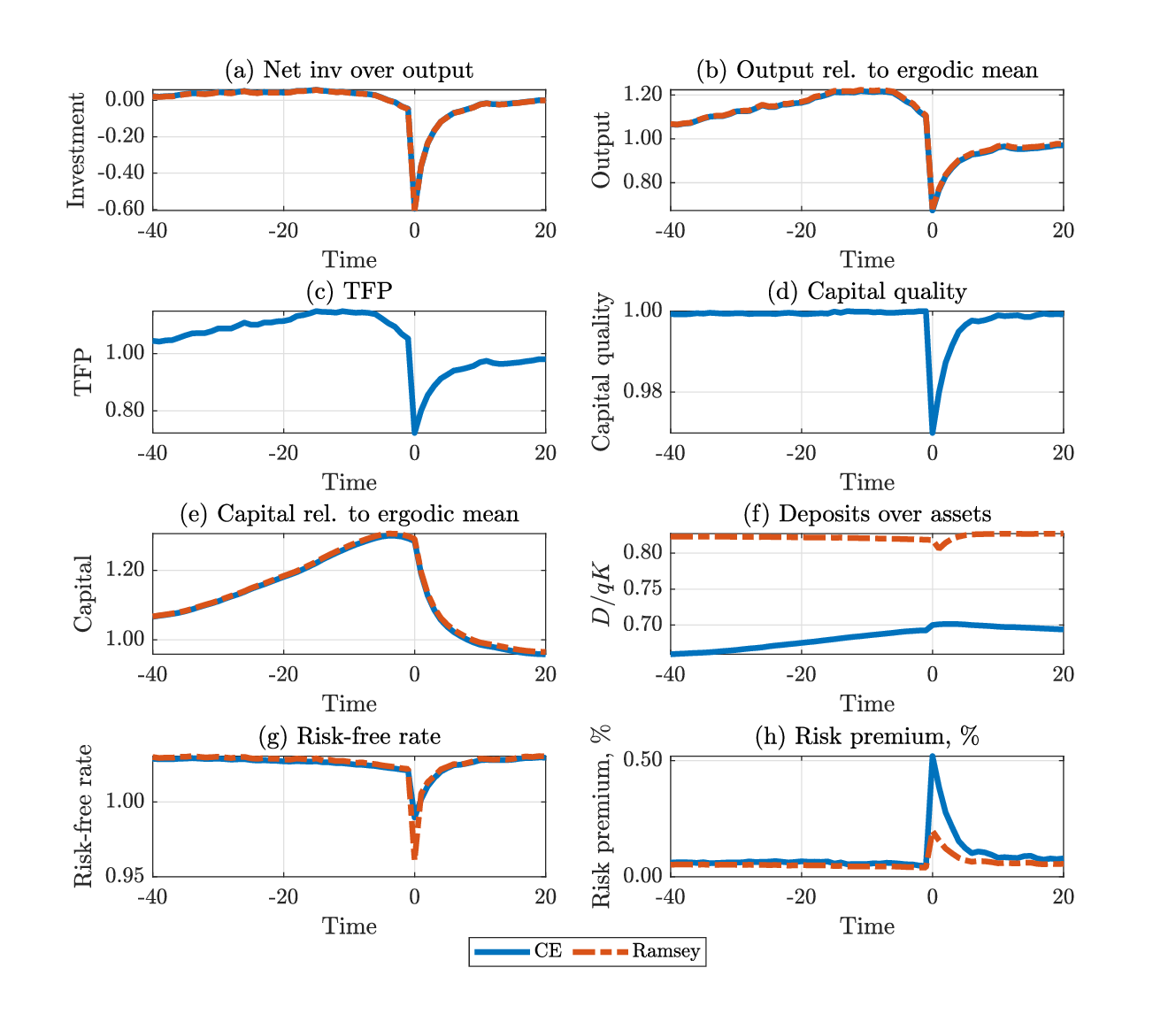}
    \caption{\label{fig:crises_ram} Crisis dynamics under Ramsey policies compared to dynamics under no policies. Panels (b) and (e) plot variables relative to their \textit{respective} ergodic means.}
\end{center}
\end{figure}

The policies and allocations described in the previous section expressed the planner's desire to redistribute between the agents. Even though similar time-0 patterns would arise in a model without exogenous shocks, the presence of shocks means that the planner needs to continually intervene to improve risk sharing. In particular, the aggregate shocks move the economy along the wealth distribution dimension by differentially affecting workers and experts. In this section, we examine the dynamics of Ramsey policies and the allocations induced by them. 

Towards that goal, Figure \ref{fig:crises_ram} repeats Figure \ref{fig:crises_eq} by overlaying the Ramsey allocations over those obtained under no policy. More precisely, the economy under Ramsey policies is simulated in the exact same way as the economy under no policies, fixing the sequence of exogenous shocks and starting with the same initial state augmented with $\mu_0=0$. Eliminating an initial burn-in period from the simulation ensures that the crisis dynamics is not affected by the strong initial redistribution motives of the planner.\footnote{I verify numerically that the dynamics under Ramsey policies is stable and that increasing the number of burn-in periods does not affect the results.} Finally, the episodes from the simulation under policy that are plotted are taken to be the same ones identified for the no-policy simulation.

\begin{figure}[h]
\begin{center}
    \includegraphics[scale = 0.70]{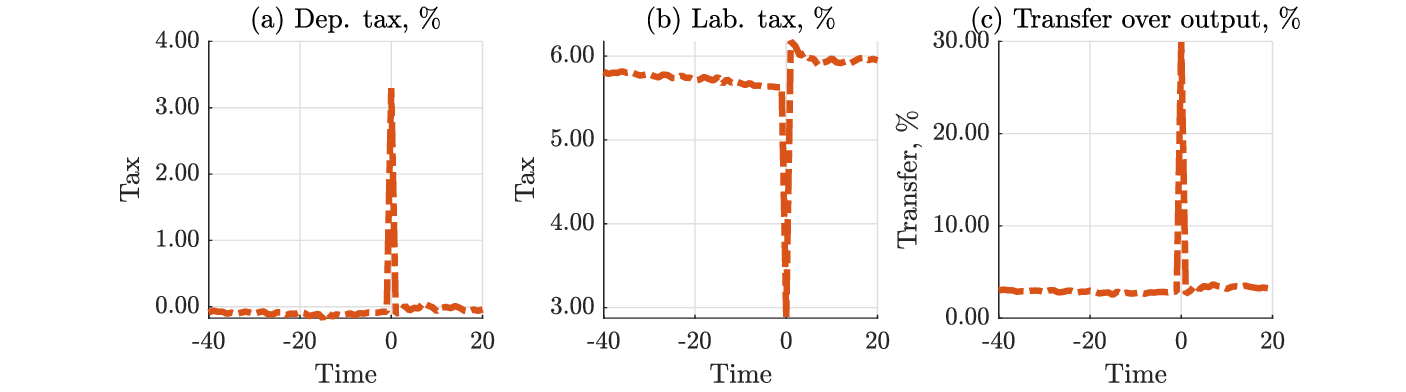}
    \caption{\label{fig:crises_tax} Dynamics of optimal taxes and transfers over crisis episodes.}
\end{center}
\end{figure}

Panel (f) demonstrates a key change in the behavior of the economy under optimal policies. The wealth of experts is substantially lower than it is under no policy for two related reasons. First, as Figure \ref{fig:crises_tax} below indicates, the planner provides some positive transfer to the experts most of the time. Second, experts respond to the insurance implicit in the policies and know that if a bad shock hits when their wealth is low, they will receive more transfers. One might be tempted to call this second effect ``moral hazard.'' However, the planner fully anticipates such effects and has the prudential instruments to counter them, if desired. Panel (a) of Figure \ref{fig:crises_tax} shows that, leading up to the crises, the tax on deposit issuance is largely inactive but plays an important role in responding to shocks. That the prudential component of policy is so muted has to do with the fact that the planner can, to a large extent, realize the gains from better risk sharing, as evident in Figure \ref{fig:risk_ramsey}. Therefore, the planner can let the economy become highly levered and respond to shocks ``ex-post.''

Panel (a) makes it clear that crises under optimal policies are of similar severity as those under no policies. It is important to note though that modest improvement does indeed take place. While net investment as a fraction of output drops about $1$ percentage points more on average when crises hit, investment under Ramsey policies recovers faster. Over the periods $10-20$, this leads to capital and output under Ramsey policies (relative to the respective ergodic means) being about $0.7$ percentage points higher than those under no policies.

Consistent with the dynamics of the tax on deposit issuance, panels (g) and (h) show the riskless rate dropping much more under Ramsey policies than under no policy. At the same time, the risk premium spikes much less, partly due to the fact that the positive deposit issuance tax increases the riskless rate faced by experts. As is apparent in panel (f), the dynamics of wealth distribution is quite different under the optimal policies. Namely, experts become relatively wealthier on impact due to redistribution and, as other panels indicate, this moves the economy to states where recovery is quicker than in the economy without taxes.

Lastly, the behavior of the labor tax shown in panel (b) of Figure \ref{fig:crises_tax} may seem puzzling. The labor tax drops on impact and immediately jumps to a higher level where it remains persistently. Such dynamics is not unusual for the labor tax in incomplete-markets environments. \citet{aiyagari2002optimal} observe similar behavior when studying the policies of a planner who needs to finance an exogenous stream of expenditures and can only trade a riskless bond with the representative agent. \citet{faraglia2019long} clarify a role for the dynamics that is present in my model as well. The decline in the labor tax on impact followed by the promise of higher rates going forward makes the workers desire more saving and, in equilibrium, helps achieve the goal of reducing the riskless rate and redistributing to the experts. In models studied by the above authors, the policy-maker wants to reduce the rate it pays on its own debt whereas in my model, the planner attempts to make borrowing effectively cheaper for the agents affected more by the negative shock. The persistence of the increase in the tax in the following periods, then, is a consequence of the desire to smooth distortions in the labor market.

\subsection{Welfare analysis \label{subsec:ramsey-welfare}}

Next, I investigate the welfare implications of implementing the optimal policies. A consumption-based welfare criterion is needed in order to make welfare across the two economies comparable. Following \citet{boar2020efficient}, I convert the objective of the planner into a constant consumption equivalent $\omega$ as follows. Given a value of the planner's objective, $V^{pl}$, $\omega$ is defined by
\begin{align*}
    V^{pl} &= \mathbb{E}_0 \left[ \sum_{t = 0}^\infty \beta^t \left(\lambda u_e(\omega) + (1-\lambda) u_w(\omega,0) \right)\right]
\end{align*}

That is, $\omega$ is the amount that, if consumed by both agents in all periods and all states of the world and with the workers not suffering any labor disutility, delivers value to the planner equal to $V^{pl}$. Given the specification of the utility functions in Section \ref{subsec:calib}, the definition of $\omega$ reduces to a simple transformation of the planner's value
\begin{align*}
    \omega &= \left((1 - \beta) (1-\gamma) V^{pl}\right) ^ {\frac{1}{1 - \gamma}}
\end{align*}

\begin{figure}[h!]
\begin{center}
    \includegraphics[scale = 0.70]{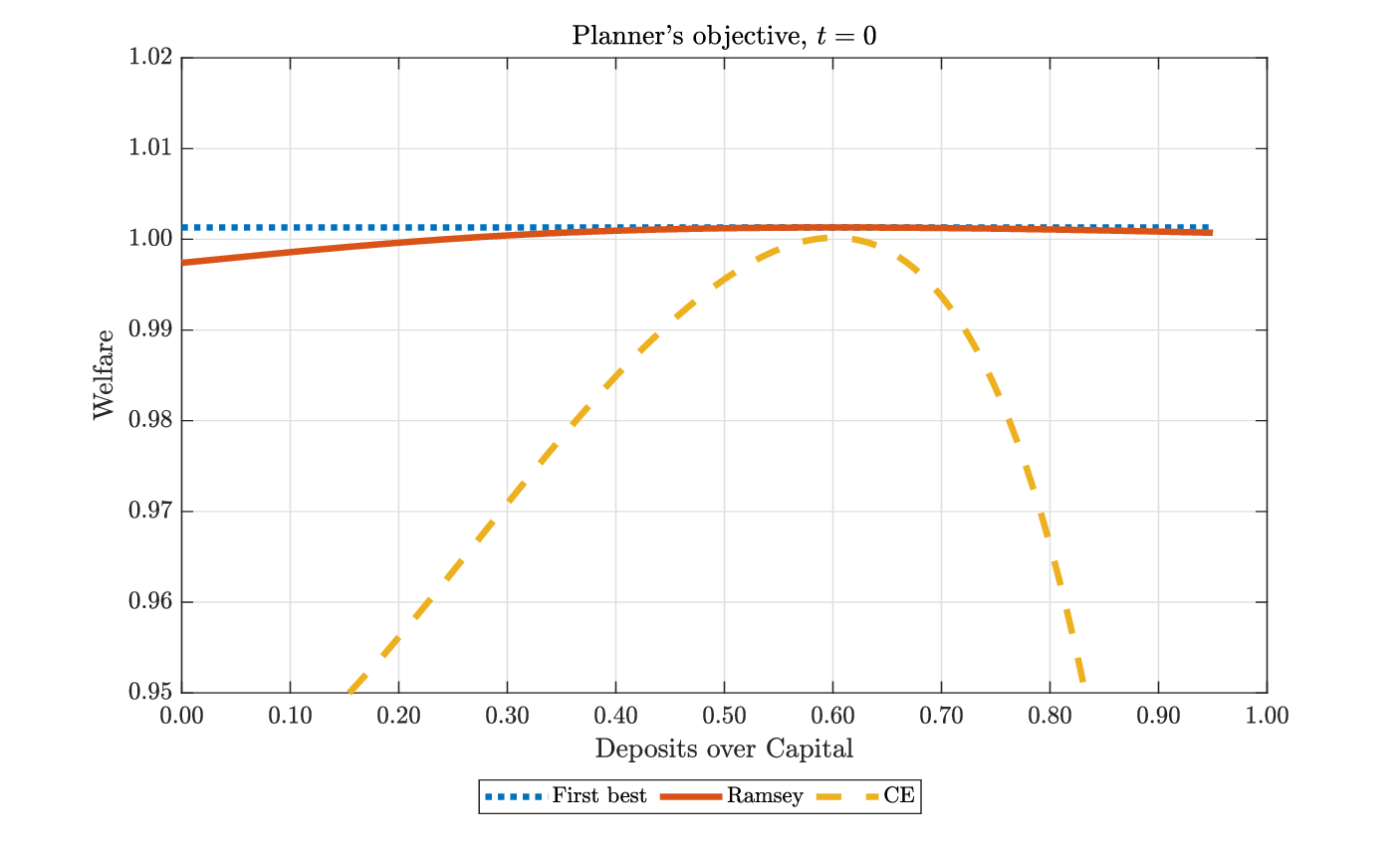}
    \caption{\label{fig:planner_welfare} Welfare at time 0 under optimal policies versus welfare under no policy. $K$ is fixed at the mean of the competitive equilibrium without policies, $\zeta=\underline{\zeta}$ and $Z=\underline{Z}$. Panel (b) shows the difference between $\zeta=\underline{\zeta}$, $Z=\underline{Z}$ and $\zeta=1$, $Z$ at median.}
\end{center}
\end{figure}

We can similarly find the consumption equivalent of the planner's value under unregulated equilibrium allocations, as well as the first-best allocations. Figure \ref{fig:planner_welfare} plots all three measures over a slice of the state space. First, note that the planner can get very close to the first-best level of welfare. Second, optimal policies are naturally most beneficial when wealth distribution is unequal. However, the welfare gains can be very modest if the initial wealth distribution is close to what the planner prefers. The competitive equilibrium without policies already features substantial risk sharing when wealth distribution is in the intermediate range (see Figure \ref{fig:risk}). Only when the exogenous shocks move the economy towards the extremes, does risk sharing become significantly worse. This fact, in conjunction with rarity of crises, explains why welfare gains are modest when wealth inequality is not extreme. Note also, that the point at which minimal welfare gains relative to the competitive equilibrium are obtained can be interpreted as the point with the planner's ``ideal'' wealth distribution. If we mute the exogenous shocks in the model, the two welfare curves will touch precisely at such a point. Therefore, the welfare gains at that point purely reflect improvement in efficiency and insurance. Furthermore, the portion of total welfare gains that are due to efficiency reasons must be increasing as we go to the extremes. In Section \ref{subsec:welfare-decomp}, I will verify this by decomposing the welfare gains into gains from efficiency, redistribution and insurance.

\subsection{Welfare decomposition \label{subsec:welfare-decomp}}

Conceptually, there are three sources for welfare gains in this model. First, optimal policy can improve aggregate efficiency by increasing expected total consumption or reducing labor. Second, the planner would like to redistribute across the agents and alter fractions of total consumption that each agent gets to enjoy. Third, policy can reduce the variance of agents' consumptions by improving insurance.

I use the methodology proposed in \cite{bhandari2021decomp} to decompose the welfare gains due to optimal policy into components that correspond to the three components above. Their method is more general than that of \cite{benabou2002tax} and \cite{floden2001effectiveness} and is more appropriate for settings like mine.

\begin{figure}[h]
\begin{center}
    \includegraphics[scale = 0.70]{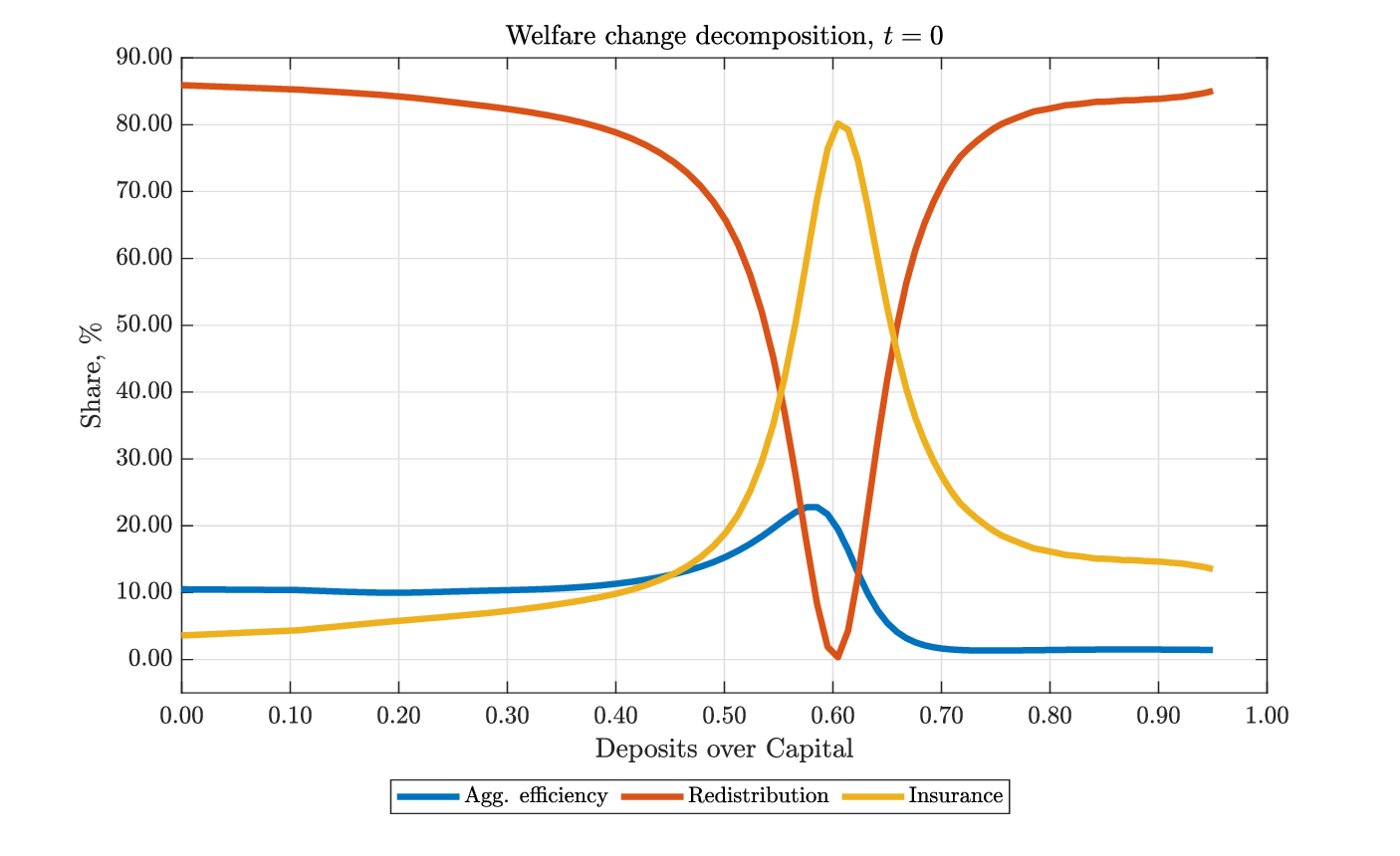}
    \caption{\label{fig:decomp_old} Welfare decomposition, time 0.}
\end{center}
\end{figure}

Figure \ref{fig:decomp_old} shows the decomposition across a slice of the state space. Gains due to redistribution account for most of the welfare gains in much of the state space and approach zero at a particular value of the deposits-to-capital ratio that depends on the Pareto weight $\lambda$. Welfare gains at that point are entirely due to improvements in efficiency and insurance. Note that, even though the fractions of gains attributable to efficiency and insurance do not rise as we go to the extremes of the state space, the absolute gains do increase (recall Figure \ref{fig:planner_welfare}). This is intuitive because risk sharing worsens in the extremes of the state space and, consequently, equilibrium labor and investment deviate from their first-best levels.

\subsection{Alternative policies \label{subsec:alter-pol}}

This section explores alternative policies inspired by what we learn about optimal policies from the previous sections.

\subsubsection{Simple rules \label{subsubsec:simple-rules}}

We found that the tax on deposit issuance finances most of the transfers in the region where the economy spends most of the time. The shape of the tax over the deposits-over-capital dimension has two distinct properties. First, the tax is increasing over most of the state space and crosses the zero line at some threshold. Second, the tax becomes decreasing when deposits are low and less redistribution is financed through it. The first property reveals that the tax on deposits aims to move the economy towards a desired wealth distribution. When the tax is positive, experts receive positive transfers and limit their deposit issuance because they perceive issuance to be more costly. Conversely, when the tax is negative, experts lose resources to workers and perceive deposit issuance to be less costly.

Motivated by these patterns, consider a simple policy rule that involves only a tax on deposits given by the following function of the state $\textbf{S}=\{K,D,Z,\zeta\}$
\begin{align*}
    \tau_d(\textbf{S}) = \tau_{d1} D \left(\frac{D}{q(\textbf{S})K}-\tau_{d2}\right)
\end{align*}

The interpretation of this quadratic form is that it targets a desirable wealth distribution or leverage given by $\tau_{d2}$. The tax function is then scaled by the quantity of deposits to express the idea that the tax on deposit issuance should decline when the tax base declines because of the inefficiency that is caused when using the tax to redistribute. Importantly, the economy -- both under no policy and under a simple policy of this type -- will not spend much time in the region with low deposits.

\begin{figure}[h!]
\begin{center}
    \includegraphics[scale = 0.70]{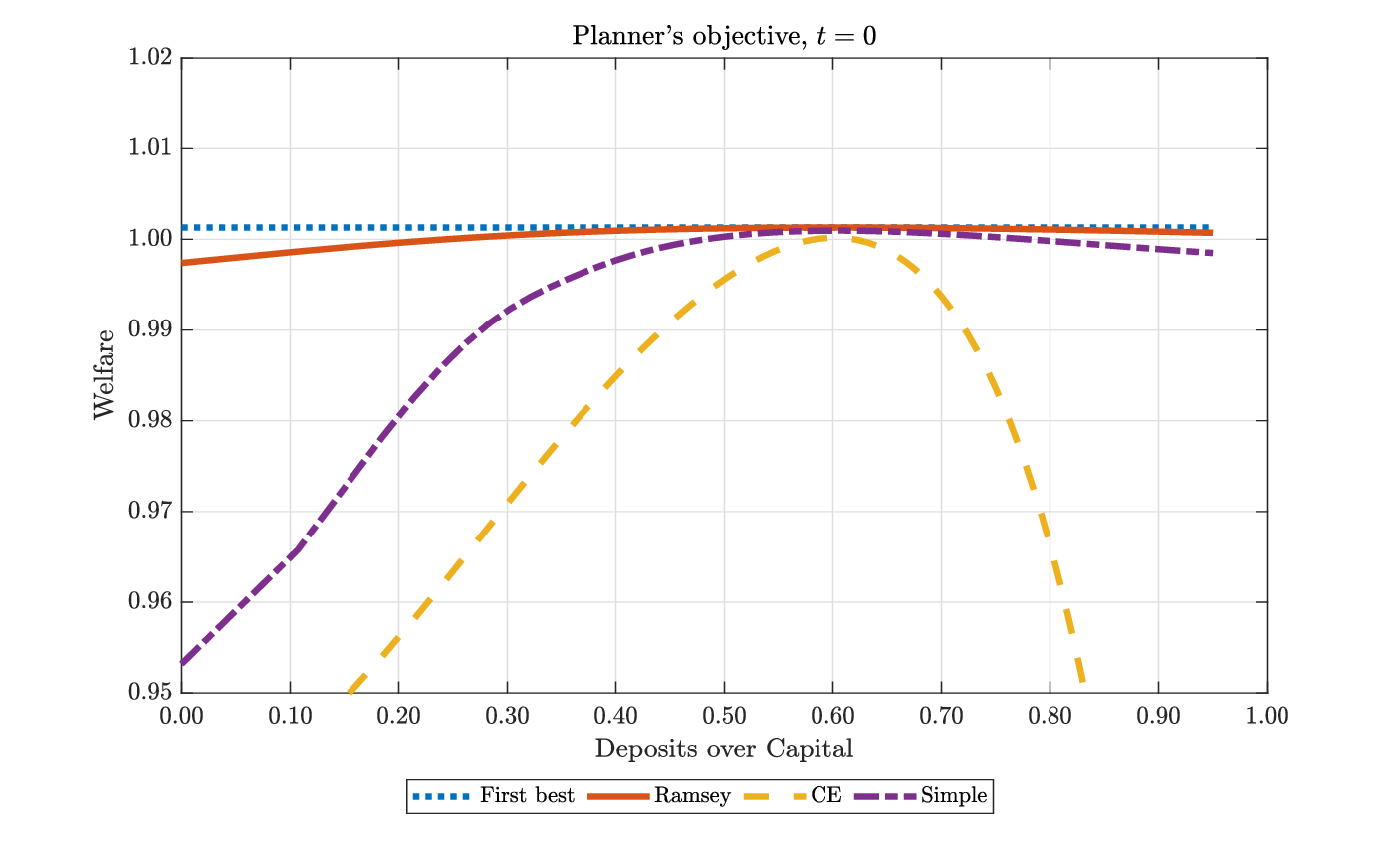}
    \caption{\label{fig:welfare_comparison} Welfare comparisons at time 0. $K$ is fixed at the mean of the competitive equilibrium without policies, $\zeta=\underline{\zeta}$ and $Z=\underline{Z}$. Panel (b) shows the difference between $\zeta=\underline{\zeta}$, $Z=\underline{Z}$ and $\zeta=1$, $Z$ at median.}
\end{center}
\end{figure}

An optimal simple rule can be found by optimizing in the space $\{\tau_{d1},\tau_{d2}\}$. The optimal policy of course depends on the initial state at which the planner optimizes.\footnote{Once again, the dependence of optimal policy on the starting point is a symptom of time-inconsistency.} For simplicity, in what follows, I suppose that a simple rule is chosen under an average value criterion. More concretely, the policy-maker evaluates a simple rule by calculating the expectation of the consumption equivalent of the planner's value over the ergodic distribution of the equilibrium under no policy. Figure \ref{fig:welfare_comparison} mirrors Figure \ref{fig:planner_welfare} and adds the consumption equivalent of the planner's value under the optimized simple rule. The simple rule attains most of the gains associated with Ramsey policies except in the region with low deposits, as expected.

The simple rule proposed here has similarities to simple policies prescribed elsewhere in the literature. \citet{bianchi2018optimal}, for instance, propose a ``macroprudential Taylor rule'' that levies a tax on the borrowing firm's debt that is zero at and below a threshold level of debt and increasing above it. The logic of that rule overlaps with that discussed here: when debt is too high, efficiency dictates lowering it. There are important differences, however. First, reducing the debt of experts in my model serves not only an efficiency but also a redistributive purpose. Second, the simple tax discussed above is symmetric and active also when expert debt is too low. Symmetricity, too, serves the dual purpose of efficiency and redistribution. Recall from Figure \ref{fig:il_ramsey} that the first-best level of investment is below the competitive level when deposits are low: experts are wealthy and do not internalize the excessive exposure of workers to the risk.


\section{Conclusion\label{sec:concl}}

This paper explored macroprudential regulation in a model with financial frictions and household heterogeneity. Inefficient risk sharing in this framework creates scope for occasional transfers to experts who invest in capital. These transfers are financed by other agents who benefit from those transfers indirectly, through increased wages in the future. Based on the mismatch in payers and direct beneficiaries of transfers, I asked to what extent prudential regulation was justified.

First, I found that prudential regulation has redistributive implications and can, by itself, be used to finance transfers. Quantitatively, transfers in crises episodes are indeed mostly financed by a tax on deposit issuance. The tax reduces the equilibrium riskless rate and effectively creates rents for the experts. Moreover, I established an analytical result according to which the planner must set the tax on capital purchases to zero. The policy-maker should not appeal to the tax on capital purchases precisely because it can control the relevant margins of redistribution using the remaining tax instruments.

Second, I found that the tax on deposit issuance is not active in the lead up to crises. The reason for this is that the optimal macroprudential policy gets close to the efficient outcomes and, in particular, equalizes risk across the agents to a substantive degree. Since risk is shared well, the agents are mostly in agreement when it comes to the amount of investment and the planner does not intervene much.

Finally, I explored the welfare implications of using optimal policies and emphasized the importance of gains due to the redistributive motive. The gains due to better risk sharing are modest when the initial state of the economy features good risk sharing but increase if the starting state is associated with high leverage. Moreover, I demonstrated that a simple macroprudential rule that can be interpreted to be targeting an ideal leverage can achieve most welfare gains across much of the state space.

A feature often present in models of macro-finance that was omitted here to preserve computational tractability is a limited participation constraint. The constraint induces a pecuniary externality and can improve the quantitative performance of the model. Relatedly, my model assumed no frictions between the finances of ``financial intermediaries'' and their owner-experts. A limited participation or equity friction would allow to decouple the finances of the firm and its owner and study episodes where owners of the firms suffer limited losses but the intermediaries they own sharply lose capacity to channel resources to investment. This is an interesting avenue for future research.

\pagebreak

\bibliographystyle{aer}
\bibliography{library}

\pagebreak

\appendix
\section{Proofs and Derivations} \label{app:proofs}

\subsection{Equilibrium conditions\label{app:eqbm-cond}}

The first-order conditions of experts' problem \eqref{value-expert}-\eqref{trans-expert} with respect to consumption and next-period deposits imply the following Euler equation
\begin{align}
    u_{ec}(c_e(k_e,d_e,\textbf{S})) = \beta \frac{r(\textbf{S})}{1-\tau_d(\textbf{S})} \mathbb{E}\left[ u_{ec}(c_e(k_e',d_e',\textbf{S}'))\right] \label{eq:ind-expert-euler}
\end{align}

Moreover, the first-order condition with respect to next-period capital implies the arbitrage condition
\begin{align}
    \mathbb{E}_t\left[u_{ec}(c_{e}(k_e',d_e',\textbf{S}'))\right]\frac{1 + \tau_{k}(\textbf{S})}{1-\tau_{d}(\textbf{S})}r(\textbf{S}) &= \mathbb{E}\left[u_{ec}(c_{e}(k_e',d_e',\textbf{S}')) \left(\frac{\zeta' \left(R(\textbf{S}') + (1 - \delta)q(\textbf{S}') + \Pi(\textbf{S}')\right)}{q(\textbf{S})} \right)\right] \label{eq:ind-arbitrage}
\end{align}

The first-order conditions of workers' problem \eqref{value-worker}-\eqref{trans-worker} similarly yield
\begin{align}
    u_{wc}(c_w(d_w,\textbf{S}), l_w(d_w, \textbf{S})) = \beta r(\textbf{S}) \mathbb{E}\left[ u_{wc}(c_w(d_w',\textbf{S}'),l_w(d_w', \textbf{S}'))\right] \label{eq:ind-worker-euler}
\end{align}

as well as the labor supply condition
\begin{align}
    -\frac{u_{wl}(c_w(d_w,\textbf{S}),l_w(d_w, \textbf{S}))}{u_{wc}(c_w(d_w,\textbf{S}),l_w(d_w, \textbf{S}))} &= (1 - \tau_{l}(\textbf{S}))W(\textbf{S}) \label{eq:ind-mrs}
\end{align}

Applying aggregation to these conditions and recognizing that individual states coincide with aggregate ones, as well as using the experts' optimal investment decision and final good producers' optimality conditions we arrive at the following system of equilibrium conditions
\begin{align}
    -\frac{u_{wl}(C_{w}(\textbf{S}),L(\textbf{S}))}{u_{wc}(C_{w}(\textbf{S}),L(\textbf{S}))} &= (1 - \tau_{l}(\textbf{S}))F_L(Z, K, L(\textbf{S})) \label{eq:eqcond-mrs} \\
    u_{wc}(C_{w}(\textbf{S}),L(\textbf{S})) &= \beta r(\textbf{S}) \mathbb{E}_t\left[ u_{wc}(C_{w}(\textbf{S}'),L(\textbf{S}')) \right] \label{eq:app-worker-euler}\\
    (1-\tau_{d}(\textbf{S}))u_{ec}(C_{e}(\textbf{S})) &= \beta r(\textbf{S}) \mathbb{E}\left[ u_{ec}(C_e(\textbf{S}')) \right] \label{eq:app-expert-euler}\\
    \mathbb{E}\left[u_{ec}(C_{e}(\textbf{S}'))\right]\frac{1 + \tau_{k}(\textbf{S})}{1-\tau_{d}(\textbf{S})}r(\textbf{S}) &= \mathbb{E}\left[u_{ec}(C_{e}(\textbf{S}')) \left(\frac{\zeta' \left(R(\textbf{S}') + (1 - \delta)q(\textbf{S}') + \Pi(\textbf{S}')\right)}{q(\textbf{S})} \right)\right] \label{eq:app-arbitrage}\\
    1 &= q(\textbf{S})\Phi'\left(\frac{I(\textbf{S})}{K}\right)\\
    Y(\textbf{S}) &= I(\textbf{S}) + C_w(\textbf{S}) + C_e(\textbf{S})\\
    K'(\textbf{S}) &= \zeta' K\left(\Phi\left(\frac{I(\textbf{S})}{K}\right)  +  (1-\delta)\right)\\
    \frac{D'(\textbf{S})}{r(\textbf{S})} + C_{w}(\textbf{S}) &= D + F_L(Z,K,L(\textbf{S})) L(\textbf{S}) (1-\tau_{l}(\textbf{S}))  \label{eq:eqcond-dlom}
\end{align}

\subsection{Proof of Proposition \ref{thm:cap-tax}\label{app:cap-tax-proof}}

Assuming interior solution, the first-order conditions of the problem \eqref{eq:ram-val}-\eqref{eq:ram-primal-bounds} with respect to the consumption of experts, as well as investment yield the following Euler equation
\begin{align}
    &u_c(C_{et}) = \beta \mathbb{E}_t\left[u_c(C_{et+1}) \zeta_{t+1}\left(F_K(Z_{t+1},K_{t+1},L_{t+1}) + (1-\delta) + \frac{\Phi\left(\frac{I_{t+1}}{K_{t+1}}\right)}{\Phi'\left(\frac{I_{t+1}}{K_{t+1}}\right)} - \frac{I_{t+1}}{K_{t+1}} \right)  \right] \Phi'\left(\frac{I_t}{K_t}\right)
\end{align}

Noting the expressions for $q_t$ and $\Pi_t$, the above can be rewritten as
\begin{align}
    &u_c(C_{et}) = \beta \mathbb{E}_t\left[u_c(C_{et+1}) \left(\frac{\zeta_{t+1}(F_K(Z_{t+1},K_{t+1},L_{t+1}) + (1-\delta) + \Pi_{t+1}}{q_t} \right)  \right]
\end{align}

The corresponding equilibrium condition (essentially, a combination of equations \eqref{eq:app-expert-euler} and \eqref{eq:app-arbitrage}) is 
\begin{align}
    &u_c(C_{et}) = \beta \mathbb{E}_t\left[u_c(C_{et+1}) \left(\frac{\zeta_{t+1}(F_K(Z_{t+1},K_{t+1},L_{t+1}) + (1-\delta) + \Pi_{t+1}}{q_t (1+\tau_{kt})} \right)  \right]
\end{align}

Comparing the previous two equations reveals that
\begin{align}
    \tau_{kt} = 0
\end{align}

\subsection{Proof of Proposition \ref{thm:inel-labor}\label{app:inel-labor-proof}}

When labor is perfectly inelastic, the constant level of equilibrium labor is given by
\begin{align}
    L^* = \lim_{\nu \to \infty} \left(\frac{(1-\tau_l) W C_w^{-\gamma}}{\psi}\right) ^{\frac{1}{\nu}} = 1
\end{align}

In particular, labor supply does not respond to the tax $\tau_l$ and the budget constraint of the worker looks like
\begin{align}
    c_w + \frac{d_w'}{r(\textbf{S})} = & d_w + \left(1 - \tau_l(\textbf{S})\right)F(Z,K,1)
\end{align}

Clearly, $\left(1 - \tau_l(\textbf{S})\right)F(Z,K,1)$ is only a function of the aggregate state and can be chosen by the planner arbitrarily
\begin{align}
    \Hat{T}(\textbf{S}) = \left(1 - \tau_l(\textbf{S})\right)F(Z,K,1)
\end{align}

Recognizing this reveals that the constraint \eqref{eq:dlom} of the planner's problem cannot be binding since any feasible value $D_{t+1}$ can be chosen by the planner by choosing the appropriate value of the transfer residually. This implies that the planner's problem reduces to the social planner's problem \eqref{eq:sp-val}-\eqref{eq:sp-rc}. Integrating the condition \eqref{eq:sp-euler} that arises from that problem
\begin{align}
    \frac{\mathbb{E}_t\left[u_{ec}(C_{wt+1})\right]}{u_{ec}(C_{wt})} = \frac{\mathbb{E}_t\left[u_{wc}(C_{wt+1},L_{t+1})\right]}{u_{wc}(C_{wt},L_t)}
\end{align}

Comparing this to equations \eqref{eq:app-worker-euler} and \eqref{eq:app-expert-euler} implies that the tax on deposit issuance must be set to zero.

Finally, since the state of the social planner is only the physical state $\textbf{S}$, the solution is time-consistent.

\section{Extension with non-unitary measures of agents} \label{app:non-unit}

\subsection{Model extension \label{app:non-unit-model}}

Consider an extension of the model where there is a measure $\theta$ of experts and a measure $1-\theta$ of workers. In this section, I show that for any value of the labor disutility scaling parameter $\psi$ and for any $\theta$ there exists a different value for the scaling parameter, $\hat{\psi}$, so that the equilibrium aggregates under $\psi$ and unitary measures coincide with the aggregates under $\hat{\psi}$ and measures $\theta$ and $1-\theta$.

Conjecture that all prices, rates and aggregate allocations are the same in the two economies. To verify that those constitute an equilibrium in the new economy, begin by the experts' Euler equation \eqref{eq:ind-expert-euler} and apply aggregation
\begin{align}
    u_{ec}\left(\frac{C_e(\textbf{S})}{\theta}\right) &= \beta \frac{r(\textbf{S})}{1-\tau_d(\textbf{S})} \mathbb{E}\left[ u_{ec}\left(\frac{C_e(\textbf{S}')}{\theta}\right)\right]
\end{align}

By assumption of the CRRA functional form, the equation reduces to
\begin{align}
    u_{ec}\left(C_e(\textbf{S})\right) &= \beta \frac{r(\textbf{S})}{1-\tau_d(\textbf{S})} \mathbb{E}\left[ u_{ec}\left(C_e(\textbf{S}')\right)\right]
\end{align}

The workers' measure is eliminated from the corresponding Euler equation in a similar fashion. The optimal investment decision holds as follows
\begin{align}
    q(\textbf{S}) = \frac{1}{\Phi'\left(\frac{I(\textbf{S})/\theta}{K/\theta}\right)} = \frac{1}{\Phi'\left(\frac{I(\textbf{S})}{K}\right)}
\end{align}

and per unit profits are the same
\begin{align}
    \Pi(\textbf{S}) = q(\textbf{S})\Phi\left(\frac{I(\textbf{S})/\theta}{ K/\theta}\right)-\frac{I(\textbf{S})/\theta}{K/\theta} =  q(\textbf{S}')\Phi\left(\frac{I(\textbf{S}')}{ K'}\right)-\frac{I(\textbf{S}')}{K'}
\end{align}

The arbitrage equation \eqref{eq:ind-arbitrage} then reduces to the equilibrium condition \eqref{eq:app-arbitrage} by the CRRA assumption.

Given the functional forms, the labor supply decision of workers, \eqref{eq:ind-mrs}, takes the form
\begin{align}
    \hat{\psi} \frac{l_w(d_e,\textbf{S})^{\nu}}{c_w(d_e,\textbf{S})^{-\gamma}} &= (1-\tau_l(\textbf{S}))W(\textbf{S}) \quad \iff \\
    \hat{\psi} (1-\theta)^{-\nu-\gamma} \frac{L(\textbf{S})^{\nu}}{C_w(\textbf{S})^{-\gamma}} &= (1-\tau_l(\textbf{S}))W(\textbf{S})
\end{align}

Therefore, to support this equilibrium, I set
\begin{align}
    \hat{\psi} = \psi (1-\theta)^{\nu+\gamma}
\end{align}

Finally, the aggregated budget constraints of the two types hold trivially and the system of equilibrium conditions from Appendix \ref{app:eqbm-cond} is obtained.

\subsection{Pareto weights \label{app:non-unit-weight}}

Under the model extension, a utilitarian planner's objective becomes
\begin{align}
    &\mathbb{E}_0 \left[ \sum_{t = 0}^\infty \beta^t \left(\theta u_e\left(\frac{C_{et}}{\theta}\right) + (1-\theta) u_w\left(\frac{C_{wt}}{1-\theta},\frac{L_t}{1-\theta};\hat{\psi}\right)  \right) \right] = \\
    &\mathbb{E}_0 \left[ \sum_{t = 0}^\infty \beta^t \left(\theta^\gamma \frac{{C_{et}}^{1-\gamma}}{1-\gamma} + (1-\theta)^{\gamma} \frac{{C_{wt}}^{1-\gamma}}{1-\gamma} -(1-\theta)^{\gamma}\psi \frac{L_t^{1+\nu}}{1+\nu} \right)  \right] =\\
    &\mathbb{E}_0 \left[ \sum_{t = 0}^\infty \beta^t \left(\theta^\gamma u_e\left(\frac{C_{et}}{\theta}\right) + (1-\theta)^{\gamma} u_w\left(\frac{C_{wt}}{1-\theta},\frac{L_t}{1-\theta};\psi\right)  \right) \right]
\end{align}

Therefore, in the baseline model, taking
\begin{align}
    \lambda = \frac{\theta^\gamma}{\theta^\gamma + (1-\theta)^\gamma}
\end{align}

makes the planner's objective equivalent to that of a utilitarian planner under the extension with non-unitary measures. Moreover, the baseline value $\lambda=0.01$ parametrized in Table \ref{tab:calib} corresponds to the value $\theta \approx 0.09$, i.e. an assumption that experts comprise $9\%$ of the population.

\section{Numerical methods} \label{app:numerical-method}

\subsection{Competitive equilibrium \label{app:num-ce}}

The numerical method for finding the equilibrium consists of finding solutions to the system of functional equations \eqref{eq:eqcond-mrs}-\eqref{eq:eqcond-dlom}. To do that, I approximate aggregate variables -- allocations, prices and rates -- using cubic splines so that for any variable $\chi$, the function $\chi(\textbf{S})$ is approximated by the function $\phi(\textbf{S})c^{\chi}$ where $\phi(\textbf{S})$ is the appropriate basis and the vector $c^{\chi}$ represents the approximation coefficients. I define the cubic splines on a collection of collocation points $\{\textbf{S}_1, \cdots,\textbf{S}_{N^{c}}\}$ in the state space. Let $\textbf{c}$ denote the matrix of the coefficients for all variables. I reduce the system of equilibrium conditions to a $3\times 3$ system and find the solution of the approximated $3N^c \times 3N^c$ system of equations denoted as follows

\begin{align}
    \mathcal{G} \left(\begin{bmatrix}
    \textbf{S}_1 \\
    \vdots \\
    \textbf{S}_{N^c}
    \end{bmatrix}, \textbf{c}\right) = 0
\end{align}

The system $\mathcal{G}$ is non-linear in $\textbf{c}$ and so a good starting guess for the solution is critical and hard to find. To find a good starting guess, I employ two techniques.

First, I use an iterative method to arrive at values $\textbf{c}_{N^I}$ where the system $\mathcal{G}$ is close to zero. The iterative method relies on economic intuition to rewrite the system $\mathcal{G}$ in a forward-looking form. Let $\textbf{c}^1$ be the matrix of coefficients describing equilibrium outcomes in the next period and $\textbf{c}^0$ represent the outcomes in the current period. Then, one can solve the rewritten system
\begin{align}
    \mathcal{H} \left(\begin{bmatrix}
    \textbf{S}_1 \\
    \vdots \\
    \textbf{S}_{N^c}
    \end{bmatrix}, \textbf{c}^0,\textbf{c}^1\right) = 0
\end{align}

for values $\textbf{c}^0$ given values $\textbf{c}^1$. Importantly, $\mathcal{H}$ is a block-diagonal system and the sparse Jacobian of the system can be calculated using the analytic expression for derivatives of cubic splines.\footnote{The derivative can likewise be computed for the system $\mathcal{G}$ but that Jacobian is not as sparse.} Using this method, one can start from some initial guess $\textbf{c}_0$ and iterate $N^I$ times to find the coefficients $\textbf{c}_{N^I}$. Although this process is relatively slow and its convergence is not theoretically guaranteed, in practice, it converges for a relatively wide range of initial $\textbf{c}_0$. With large enough $N^I$, the obtained values $\textbf{c}_{N^I}$ make the system $\mathcal{G}$ close enough to zero so that a very fast Newton-type method becomes applicable.

Second, to facilitate finding a good initial guess $\textbf{c}_0$, I use a homotopy approach. The idea of the approach is to construct a sequence of parametrizations -- consisting of sequences of values $\{\sigma_{\varepsilon 0},\cdots,\sigma_{\varepsilon N^h}\}$, $\{\gamma_{0},\cdots,\gamma_{N^h}\}$ and $\{\xi_{0},\cdots,\xi_{N^h}\}$ -- so that a good guess for the equilibrium is easy to obtain for the initial parametrization and values for $N^h$ coincide with the baseline parametrization. Then, the equilibrium found for each parametrization can be used as a starting guess for the equilibrium under the next parametrization. In particular, I start from values $\sigma_{\varepsilon 0} \approx 0$, $\gamma_0 = 1$ and $\xi = 1$. This procedure allows to find the equilibrium under the baseline parametrization relatively quickly and reliably.

\subsection{Ramsey problem \label{app:num-ramsey}}

The saddle-point problem \eqref{eq:saddle-begin}-\eqref{eq:saddle-end} can be solved using standard value function iteration methods. I approximate the value function, as well as its expectation using cubic splines, solve the stage saddle-point problem at each iteration and update. A notable practical complication that arises here is the difficulty of reliably solving the stage problem at a relatively large number of grid points. I derive a sparse, block-diagonal system of first-order conditions that describe the solution to the stage problem at all grid points. This system again features substantial non-linearity and having a good initial guess is critical. When iterating, it suffices to choose as a starting guess for the solution the one found for the previous iteration. However, extra steps are needed to find a good starting guess at the first iteration, i.e. for the initial guesses of the value function and its expectation.

To find this initial guess I apply a homotopy method similar to the one used in the previous section. In particular, owing to Proposition \ref{thm:inel-labor}, it is straightforward to construct good guesses for the solution of the planner's problem when $\nu$ is very large. To that end, consider a sequence $\{\nu_0,\cdots, \nu_{N^{\nu}}\}$ such that $\nu_0=1000$ and $\nu_{N^{\nu}}$ is the value for the baseline parametrization. Then, at each iteration $j$, the stage problem is solved assuming $\nu = \nu_j$ if $j\leq N^\nu$ and $\nu = \nu_{N^\nu}$ otherwise. For large enough $N^\nu$, this procedure ensures that a good starting guess is available at each iteration. Finally, note that the sparse Hessian of the stage problem can again be conveniently calculated using analytic expressions.

\end{document}